\documentclass[12pt]{amsart}

\usepackage{booktabs}
\usepackage{xcolor}
\usepackage{braket}
\usepackage[margin=1in]{geometry}
\usepackage{mathrsfs}
\usepackage{latexsym}
\usepackage{amssymb,amsfonts}
\usepackage[numeric,lite,initials]{amsrefs}
\definecolor{darkgreen}{rgb}{0,0.65,0}
\usepackage{graphicx}
\usepackage{wrapfig}
\usepackage{fancybox}
\usepackage{bm}
\usepackage{stmaryrd}
\usepackage{amsmath}
\usepackage{mathdots}
\usepackage{amsthm}
\usepackage{tikz-cd}
\usepackage{thmtools}
\usepackage{url}
\usepackage{color}
\usepackage[
    colorlinks=true,
    linkcolor=blue,
    citecolor=darkgreen,
    urlcolor=magenta
]{hyperref}
\usepackage{cleveref}

\declaretheorem[style=plain, numberwithin=section]{theorem}
\declaretheorem[style=definition,name=Definition,qed=$\blacksquare$, numberwithin=section, sibling=theorem]{definition}

\newtheorem{lemma}[theorem]{Lemma}

\newtheorem{proposition}[theorem]{Proposition}
\newtheorem{corollary}[theorem]{Corollary}



\newcommand{\Tr}{\operatorname{Tr}}
\newcommand{\SL}{\operatorname{SL}}
\newcommand{\GL}{\operatorname{GL}}
\newcommand{\SU}{\operatorname{SU}}
\newcommand{\Or}{\operatorname{O}}
\newcommand{\U}{\operatorname{U}}
\newcommand{\SO}{\operatorname{SO}}

\newcommand{\CC}{\mathbb C}

\newcommand{\RR}{\mathbb R}
\newcommand{\ZZ}{\mathbb Z}
\newcommand{\im}{\mathrm i}

\title[Classification of five-qubit absolutely maximally entangled states]{Classification of five-qubit absolutely maximally entangled states}

\author{Ian Tan}
\thanks{Department of Mathematics and Statistics, Auburn University, Auburn, AL (\href{mailto:yzt0060@auburn.edu}{yzt0060@auburn.edu})}

\begin{document}
\begin{abstract}
We classify the local unitary equivalence classes of absolutely maximally entangled (AME) states of five qubits. We show that every 5-qubit AME state is equivalent to a state within the unique $((5,2,3))$ quantum error-correcting code $\mathcal{C}$, and that two such states are equivalent if and only if they are related by the action of a transversal gate of $\mathcal{C}$. Furthermore, we exhibit a set of three invariant polynomials that separates these equivalence classes. 

As auxiliary results, we construct a 3-uniform $n$-qubit state for even $n\geq 6$, determine the local symmetries of the 6-qubit AME state, and explain how these symmetries are related to the transversal gates of both the $((5,2,3))$ and $((4,4,2))$ codes. Additionally, we demonstrate that every 4-qubit pure code of distance 2 is equivalent to a subspace of a $((4,4,2))$ code. 

Our approach leverages an embedding of the 4-qubit state space into the Lie algebra of $8\times 8$ skew-symmetric matrices, allowing us to apply results from Vinberg's theory of graded Lie algebras.
\end{abstract}
\maketitle

\section{Introduction}
The entanglement of multipartite quantum states plays a central role in many quantum information processing tasks. Thus, it is crucial to understand the nature of entanglement. Among the most highly entangled multipartite states are the absolutely maximally entangled (AME) states. These AME states have applications in quantum secret sharing \cite{PhysRevA.86.052335}, open-destination teleportation \cite{Helwig:2013qoq}, quantum information masking \cite{PhysRevA.104.032601}, and quantum error correction \cites{Raissi_2018,Huber2020}. A closely related and more general notion is that of an $r$-uniform state, which are states in $(\CC^D)^{\otimes n}$ for which every reduction to $r$ subsystems is maximally mixed. By Schmidt decomposition, $r$-uniform states do not exist when $r>\lfloor \frac{n}{2}\rfloor$. An AME state is an $r$-uniform state such that $r=\lfloor \frac{n}{2}\rfloor$.

A longstanding research problem is to determine for which multipartite systems AME states exist \cites{Huber_2018,AMEtable}. In this pursuit, one shows the existence of AME states by construction \cites{Helwig:2013qoq,PhysRevA.92.032316} or the nonexistence often through bounds that are violated \cites{Huber_2018,PhysRevA.69.052330,10718358,796376}. After establishing the existence of AME states in a state space, what else can be said? For 3-partite systems, all AME states are critical and a lot is known about such states \cites{Bryan2019locallymaximally,Slowik2021,Bryan2018}. Outside of this relatively simple situation, it is known that the state spaces $(\CC^2)^{\otimes 6}$ and $(\CC^3)^{\otimes 4}$ each contain only one normalized AME state up to local unitary equivalence \cites{746807,PhysRevA.108.032412}. For some systems, we know that there exist infinitely many inequivalent AME states \cites{Burchardt,PhysRevA.108.032412,Ramadas_2025}.

In this paper, we give a complete classification of AME states of five qubits. Specifically, we prove the following (in \Cref{thm:mainresult} and elsewhere in this paper, $X$ and $Z$ denote the Pauli matrices):
\begin{theorem}[Main theorem]\label{thm:mainresult}
    Every AME state in $(\CC^2)^{\otimes 5}$ is $\SU_2^{\times 5}$-equivalent to a point of the $((5,2,3))$ code $\mathcal{C}_1$. Two points in $\mathcal{C}_1$ are $\SU_2^{\times 5}$-equivalent if and only if they are $W$-equivalent, where $W\cong\langle\im X,\im Z,H\rangle$ is the binary tetrahedral group and
    \[
H:=\frac{1}{2}\begin{pmatrix}
    1+\im & 1+\im \\
    -1+\im & 1-\im
\end{pmatrix}.
    \]
    Futhermore, there exist three $\SL_2^{\times 5}$-invariant polynomials $\mathscr{F}_6,\mathscr{F}_8,\mathscr{F}_{12}$ of degrees $6$, $8$, and $12$ such that two AME states $\ket{\varphi},\ket{\psi}\in(\CC^2)^{\otimes 5}$ are $\SU_2^{\times 5}$-equivalent if and only if $\mathscr{F}_i(\ket{\varphi})=\mathscr{F}_i(\ket{\psi})$ for all $i\in\{6,8,12\}$.
\end{theorem}

Much is known about the entanglement classes of 3-qubit states \cite{PhysRevA.62.062314}. In particular, every normalized 3-qubit AME state is equivalent to the GHZ state $\frac{1}{\sqrt{2}}(\ket{000}+\ket{111})$ \cite{Bryan2019locallymaximally}. Moreover, it is well-known that AME states of four qubits do not exist \cite{HIGUCHI2000213}. Thus, five is the smallest number of qubits which presents a genuine challenge. \Cref{thm:mainresult} is the first result of its kind: it is the first complete classification of AME states in a state space where there exists more than one class of AME states and where critical states are not always AME.

The proof of \Cref{thm:mainresult} is achieved by employing deep connections between $r$-uniform states, quantum error correcting codes, and Vinberg's theory of graded Lie algebras. Briefly, the $((6,1,4))$ code, the $((5,2,3))$ code $\mathcal{C}_1$, and the $((4,4,2))$ code $\mathcal{C}_2$ are all unique up to local unitary equivalence; together, they constitute a family of quantum error correcting codes related by a construction of Rains \cites{Rains1996QuantumWE,746807}. The span of any 6-qubit AME state is a $((6,1,4))$ code and the local symmetries of a 6-qubit AME state determine the transversal gates of the codes $\mathcal{C}_1$ and $\mathcal{C}_2$. The code $\mathcal{C}_2$ is also a Cartan subspace of the $\SL_2^{\times 4}$-module $(\CC^2)^{\otimes 4}$ viewed as the grade-1 subspace of the $\ZZ_2$-graded Lie algebra
$
\mathfrak{so}_8\cong\mathfrak{so}_4^{\times 2}\oplus (\CC^2)^{\otimes 4}.
$
Furthermore, the normalizer of $\mathcal{C}_2$ as a Cartan subspace is the group of transversal gates of $\mathcal{C}_2$ as a quantum error correcting code.

The fact that $(\CC^2)^{\otimes 4}$ lives in the graded Lie algebra $\mathfrak{so}_8$ has many consequences in quantum information theory. For example, it was used in previous work to determine the equivalence classes of 4-qubit AME states related by stochastic local operations with classical communication \cites{ChtDjo:NormalFormsTensRanksPureStatesPureQubits,dietrich2022classification}. It has also been used to find 4-qubit states that are stationary points of entanglement measures arising from $\SL_2^{\times 4}$-invariants \cite{4qubit}. We emphasize in particular the following, which is comparable to \Cref{thm:mainresult}.

\begin{theorem}[Classification of 4-qubit 1-uniform states]\label{thm:wallach1u}
    Every 1-uniform state in $(\CC^2)^{\otimes 4}$ is $\SU_2^{\times 4}$-equivalent to a point of the $((4,4,2))$ code $\mathcal{C}_2$. Two points in $\mathcal{C}_2$ are $\SU_2^{\times 4}$-equivalent if and only if they are $W$-equivalent, where $W$ is the Weyl group of the Lie algebra $\mathfrak{so}_8$.
\end{theorem}
Although \Cref{thm:wallach1u} is an easy consequence of standard results from Vinberg's theory and geometric invariant theory, it is not usually stated in the language of quantum information. We will provide a proof of \Cref{thm:wallach1u} later in the text for the reader's convenience, for comparison with the proof of \Cref{thm:mainresult}, and because it does not take much additional effort. It is also known that the set of $\SL_2^{\times 4}$-invariant polynomials on $(\CC^2)^{\otimes 4}$ is a free algebra generated by four polynomials of degrees 2, 6, 8, and 12 \cite{GourWallach2014}. The $\SU_2^{\times 4}$-orbits of 4-qubit 1-uniform states can be distinguished by their values on these generators. In general, the $\SU_D^{\times n}$-orbits of $r$-uniform states in $(\CC^D)^{\otimes n}$ are separated by the set of all $\SL_D^{\times n}$-invariant polynomials \cite{PhysRevLett.111.060502}. However, there are many $\SL_D^{\times n}$-invariant polynomials and they are generally difficult to compute. For five qubits, finding the complete set of secondary $\SL_2^{\times 5}$-invariants seems beyond reach \cite{Luque_2006}. It is thus noteworthy that we exhibit three $\SL_2^{\times 5}$-invariant polynomials that separate the $\SU_2^{\times 5}$-orbits of AME states in $(\CC^2)^{\otimes 5}$. We find it mathematically interesting that \Cref{thm:mainresult} extends phenomena of the $\SL_2^{\times 4}$-module $(\CC^2)^{\otimes 4}$ to the $\SL_2^{\times 5}$-module $(\CC^2)^{\otimes 5}$ even though $(\CC^2)^{\otimes 5}$ is not naturally embedded in a graded Lie algebra. Note that the groups mentioned in \Cref{thm:mainresult,thm:wallach1u} are precisely the groups of operators on $\mathcal{C}_1$ and $\mathcal{C}_2$ implemented by the transversal gates (as defined in \Cref{sec:stab}) of these codes.

On the way to proving \Cref{thm:mainresult}, we achieve several other notable results, listed below. These further highlight the interconnections revealed in this work.
\begin{itemize}
    \item[1.] We construct a family of 3-uniform states $\ket{\Psi_m}\in (\CC^2)^{\otimes 2m}$ for $m\geq 3$. In contrast to other constructions of $r$-uniform states \cites{PhysRevA.90.022316,PhysRevA.97.062326,PhysRevA.99.042332,Pang2019}, which are often combinatorial, ours arises from a change of coordinates related to the double cover $\SL_2\times\SL_2\to \SO_4$. In previous work, Zang et al. \cite{ZangTianFeiZuo} construct 3-uniform states in $(\CC^2)^{\otimes n}$ for all $n\geq 6$ and $n\neq 7,8,9,11$. Note that 3-uniform states of seven qubits do not exist \cite{PhysRevLett.118.200502}.
    \item[2.] We calculate the local invertible operations that fix the 6-qubit AME state, also known as local symmetries. By a result of Spee et al., the local symmetries of a state tell us about the reachability of elements in its $\SL_D^{\times n}$-orbit: it gives a characterization of the pure states that can be converted to a given state in the orbit via finitely many rounds of local operations assisted with classical communication \cite{Spee_2017}. We show that these local symmetries determine the transversal gates of $\mathcal{C}_1$ and are closely related to the transversal gates of $\mathcal{C}_2$. Englbrecht and Kraus have a general algorithm for finding the local symmetries of a graph state \cite{PhysRevA.101.062302}, of which the 6-qubit AME state is an example \cite{Helwig:2013qoq}. However, they do not describe the local symmetries of the $6$-qubit AME state explicitly, nor do they mention the connection to transversal gates.
    \item[3.] We show that every pure code in $(\CC^2)^{\otimes 4}$ is a subspace of a $((4,4,2))$ code. This is analogous to the Vinberg theoretical fact that every abelian subspace consisting of semisimple elements is a subspace of a Cartan subspace.
\end{itemize}

We cover preliminaries in \Cref{sec:prelims}, including content on critical states, $r$-uniform states, and quantum error correcting codes. We also explain a certain construction of new codes from old ones due to Rains; this process allows one to generate the $((5,2,3))$ and $((4,4,2))$ codes from the 6-qubit AME state.

In \Cref{sec:three}, we study various aspects of the 6-qubit AME state. In \Cref{sec:3uniform}, we show that the AME state belongs to a family of 3-uniform states with members in $(\CC^2)^{\otimes n}$ for each even $n\geq 6$. In \Cref{sec:stab}, we calculate the local symmetries of the AME state and explain how they relate to the transversal gates of the codes it generates.

In \Cref{sec:vinberg}, we explain how $(\CC^2)^{\otimes 4}$ is embedded in the $\ZZ_2$-graded Lie algebra $\mathfrak{so}_8$. This embedding allows us to apply various results from Vinberg's theory of graded Lie algebras. We explain that the transversal gates of $\mathcal{C}_2$ implement the Weyl group of $\mathfrak{so}_8$, prove \Cref{thm:wallach1u}, and show that every nontrivial pure code in $(\CC^2)^{\otimes 4}$ is a subspace of a $((4,4,2))$ code.

In \Cref{sec:classification}, we prove our main result, \Cref{thm:mainresult}. An explicit description of the three invariant polynomials that separate the $\SU_2^{\times 5}$-orbits of 5-qubit AME states can also be found in this section.

Some of our proofs rely on symbolic computations carried out using the computer algebra system Macaulay2 \cite{M2} (in particular, the proofs of \Cref{thm:3uniform}, \Cref{thm:first}, \Cref{thm:surj}, and \Cref{prop:invariantcomputation}). Refer to the ``code availability'' section at the end of the paper for access to the accompanying Macaulay2 files. We also included a file for symbolic evaluation of the three separating invariants of \Cref{thm:mainresult}.

In the appendix, standard definitions and facts from algebraic geometry (Appendix~\ref{sec:algebraic_geometry}) and invariant theory (Appendix~\ref{sec:invariant_theory}) are found.

\section{Preliminaries}\label{sec:prelims}

\subsection{Notation and conventions} In this section, we establish the notation, conventions, and mathematical terminology used throughout the paper.

\subsubsection{Miscellaneous}
In this paper, a ``state'' is simply a nonzero vector in $(\CC^D)^{\otimes n}$. We write $X$, $Y$, and $Z$ to denote the Pauli matrices.  If $v$ is a vector, $\bar{v}$ is its complex conjugate. The symmetric group on $m$ elements is $\mathfrak{S}_m$.

\subsubsection{Group representations}
Let $G$ be a group. A \textit{$G$-module} is a pair $(V,\rho)$, where $V$ is a complex vector space and $\rho:G\to \GL(V)$ is a representation of $G$. For simplicity, we also call $V$ a $G$-module, since the representation $\rho$ will be clear from context. We use a period to denote a group element acting on a vector, so $g.v=\rho(g)v$ for all $g\in G$ and $v\in V$. A $G$-module $(V,\rho)$ and a $H$-module $(W,\sigma)$ are \textit{equivalent} if there exist an invertible map $F:V\to W$ and an isomorphism $\phi:\rho(G)\to \sigma(H)$ such that $F(\rho(g)v)=\phi(\rho(g))F(v)$ for all $g\in G$ and $v\in V$.

Suppose $G$ is a subgroup of $\GL(\CC^D)$, so $\CC^D$ is a $G$-module. Then $(\CC^D)^{\otimes n}$ is naturally a $G^{\times n}$-module, where the representation $\rho:G^{\times n}\to\GL((\CC^D)^{\otimes n})$ is defined by
\[
\rho(g_1,\dots,g_n)=g_1\otimes\dots\otimes g_n,\quad g_i\in G.
\]
We write $G^{\otimes n}=\rho(G^{\times n})$ to denote the image of this representation.

\subsubsection{Lie groups and Lie algebras} We assume that a Lie group or Lie algebra is complex unless otherwise indicated. For example, $\SL_D=\SL_D(\CC)$ is a complex Lie group while $\SL_D(\RR)$ is a real Lie group.

\subsubsection{Tensor contraction}
Let $\ket{\varphi}\in V_1\otimes V_2\otimes\dots\otimes V_n$, where each $V_i\cong\CC^D$. As a tensor, $\ket{\varphi}$ corresponds to a linear map $V_1^*\to V_2\otimes\dots\otimes V_n$. We write $\bra{\psi}_1\ket{\varphi}$ for the image of $\bra{\psi}\in V_1^*$ under this map. Note that $\ket{\varphi}=\sum_i \ket{i}\ket{\varphi_i}$, where $\ket{\varphi_i}=\bra{i}_1\ket{\varphi}$ for $i=0,1,\dots,D-1$.

\subsection{Critical states} We discuss some properties of an important class of states known as critical states. First, let us recall the definition of an $r$-uniform state.
\begin{definition}
    A state $\ket{\varphi}\in(\CC^D)^{\otimes n}$ is $r$-\textit{uniform} if $\Tr_S(\ket{\varphi}\bra{\varphi})\propto I$ for every subset $S\subset\{1,2,\dots,n\}$ such that $|S|=n-r$.
\end{definition}
As mentioned in the introduction, if $\ket{\varphi}$ is $r$-uniform, then $r\leq \lfloor \frac{n}{2}\rfloor$. An $r$-uniform state is \textit{absolutely maximally entangled} or AME if $r= \lfloor \frac{n}{2}\rfloor$. On the other extreme, we say that $\ket{\varphi}$ is \textit{critical} if it is 1-uniform. Since the reduction of a maximally mixed state is again maximally mixed, an $r$-uniform state is $r'$-uniform for every $1\leq r'\leq r$. In particular, all $r$-uniform states are critical.

\begin{proposition}\label{prop:equiv}
    A state $\ket{\varphi}\in(\CC^D)^{\otimes n}$ is critical if and only if $\bra{\varphi}E\ket{\varphi}=0$ for every $E$ in the Lie algebra of $\SL_D^{\otimes n}$.
\end{proposition}
\begin{proof}
    This characterization of critical states is taken to be the definition in \cite{Gour_2011}. It is equivalent to our definition by \cite[Theorem 3]{Gour_2011}.
\end{proof}

\begin{proposition}[The Kempf-Ness theorem, {\cite[Theorem~2]{Gour_2011}}]\label{prop:K-N}
Let $\ket{\varphi}\in(\CC^D)^{\otimes n}$. The following are true:
\begin{itemize}
    \item[1.] $\ket{\varphi}$ is critical if and only if $\bra{\varphi}g^\dagger g \ket{\varphi}\geq \braket{\varphi|\varphi}$ for all $g\in \SL_D^{\otimes n}$.
    \item[2.] If $\ket{\varphi}$ is critical and $g\in \SL_D^{\otimes n}$ such that $\bra{\varphi}g^\dagger g \ket{\varphi}=\braket{\varphi|\varphi}$, then there exists $h\in \SU_D^{\otimes n}$ such that $g\ket{\varphi}=h\ket{\varphi}$.
    \item[3.] The $\SL_D^{\times n}$-orbit of $\ket{\varphi}$ is closed if and only if it contains a critical state.
\end{itemize}
\end{proposition}

Note the following consequence of the Kempf-Ness theorem: two critical states are $\SL_D^{\times n}$-equivalent if and only if they are $\SU_D^{\times n}$-equivalent. This explains why $\SL_2^{\times 5}$-invariants separate $\SU_2^{\times 5}$-orbits of 5-qubit AME states in \Cref{thm:mainresult}.

\subsection{Quantum error correcting codes}\label{sec:QECC} Quantum error correcting codes are subspaces $\mathcal{C} \subset (\mathbb{C}^D)^{\otimes n}$ designed to detect and correct local errors. The state of a single $K$-level system is encoded into a larger $n$-partite system by embedding it into the subspace $\mathcal{C}$, where $K=\dim\mathcal{C}$. Our interest in quantum error-correcting codes stems from their close connection with $r$-uniform states: in particular, an especially well-behaved class of codes, known as pure codes, is equivalent to subspaces consisting of $r$-uniform states.

In order to formally state the definition of a quantum error correcting code, we need to introduce some preliminary notions. Let $\mathcal{B}$ be an orthogonal basis of $\CC^{D\times D}$ such that $I\in\mathcal{B}$. From $\mathcal{B}$, we obtain the basis
$
\mathcal{E}=\{E_1\otimes\dots\otimes E_n:E_i\in\mathcal{B},\: 1\leq i\leq n\}
$
of the space of linear maps $(\CC^D)^{\otimes n}\to (\CC^D)^{\otimes n}$. The elements of $\mathcal{E}$ represent errors that a code tries to correct. The \textit{weight} $\text{wt}(E)$ of an error $E=E_1\otimes\dots\otimes E_n\in\mathcal{E}$ is equal to the number $|\{i:E_i\neq I\}|$ of single-qudit subsystems on which $E$ acts nontrivially.

\begin{definition}
    Let $\mathcal{C}\subset (\CC^D)^{\otimes n}$ be a $K$-dimensional subspace with orthogonal basis $\{\ket{\varphi_i}\}_{i=1}^K$. The subspace $\mathcal{C}$ is a \textit{quantum error correcting code} (or simply \textit{code}) with parameters $((n,K,d))_D$ if $\bra{\varphi_i}E\ket{\varphi_j}=c(E)\delta_{ij}$ whenever $E\in\mathcal{E}$ such that $\text{wt}(E)<d$, where $c(E)$ is a complex constant depending on $E$. A code $\mathcal{C}$ is \textit{pure} if $c(E)=0$ whenever $0<\text{wt}(E)<d$.
\end{definition}

If $D=2$, we drop the subscript and simply write $((n,K,d))$. The parameter $d$ is the \textit{distance} of the code and represents the amount of error that the code can correct: a code with distance $d\geq 2t+1$ can correct errors that affect up to $t$ subsystems.
By definition, any subspace of $(\CC^D)^{\otimes n}$ is a code of distance 1. We say that a code is \textit{nontrivial} if it has distance greater than 1.

The parameters $((n,K,d))_D$ of a code satisfy an inequality known as the quantum Singleton bound: $\log_D K\leq n-2(d-1)$ \cite[Theorem 13]{Huber2020}. Any code that saturates this bound is pure.

\begin{definition}
    A code with parameters $((n,K,d))_D$ such that $\log_D K= n-2(d-1)$ is called \textit{maximum distance separable} (MDS).
\end{definition}
\begin{proposition}[{\cite[Theorem 5]{Huber2020}}]\label{thm:QMDS-pure}
    Every MDS code is pure.
\end{proposition}
\begin{proposition}[{\cite[Observation 1]{Huber2020}}]\label{prop:obs}
A $K$-dimensional subspace $\mathcal{C}\subset (\CC^D)^{\otimes n}$ is a nontrivial pure code with parameters $((n,K,d))_D$ if and only if every element of $\mathcal{C}$ is $(d-1)$-uniform.
\end{proposition}

\subsection{New codes from old ones} We now present a method for generating new pure codes from old ones. The construction was described by Rains for qubit codes \cite{Rains1996QuantumWE} and later more generally for qudits by Huber and Grassl \cite{Huber2020}.

\begin{proposition}\label{prop:rains}
    Let $\{\ket{\varphi_i}\}_{i=1}^K$ be an orthogonal basis of a pure $((n,K,d))_D$ code with $n,d\geq 2$. The operator
    \[
    D\cdot\Tr_{\{1\}}(\ket{\varphi_1}\bra{\varphi_1}+\dots + \ket{\varphi_K}\bra{\varphi_K})
    \]
    is the projection onto a pure $((n-1,DK,d-1))_D$ code.
\end{proposition}
\begin{proof}
    See the proof of \cite[Theorem 2]{Huber2020}. Note that the scalar factor $D$ is necessary to make the operator a projection. In general, the trace of a projection is equal to the dimension of its image.
\end{proof}

Suppose $\ket{\varphi}\in (\CC^D)^{\otimes n}$ is an AME state where $n$ is even. By \Cref{prop:obs}, the span of $\ket{\varphi}$ is a $((n,1,\frac{n}{2}+1))_D$ code. Applying \Cref{prop:rains} iteratively, we obtain a family of MDS codes with the following parameters.
\[
((n,1,\frac{n}{2}+1))_D,\quad ((n-1,D,\frac{n}{2}))_D, \quad ((n-2,D^2,\frac{n}{2}-1))_D,\quad \dots
\]
This process cannot always be reversed: it could be the case that a $((n-i,D^i,\frac{n}{2}-i+1))_D$ MDS code exists but no $((n-i+1,D^{i-1},\frac{n}{2}-i+2))_D$ MDS code exists (see also \cite[Section 7]{Huber2020}). However, Huber and Grassl show that the construction can always be reversed at the top two levels of an MDS family arising from an AME state. That is, the existence of a $((n-1,D,\frac{n}{2}))_D$ code implies the existence of a $((n,1,\frac{n}{2}+1))_D$ code \cite[Proposition 7]{Huber2020}.

\section{Local symmetries of the 6-qubit AME state}\label{sec:three}

\subsection{MDS family from the 6-qubit AME state}\label{sec:main} This paper focuses on the MDS family generated by the 6-qubit AME state. Thus, the MDS family consists of quantum error correcting codes with parameters $((6,1,4))$, $((5,2,3))$, and $((4,4,2))$. Every code in this family is essentially unique; we state this fact as a proposition for future reference.

\begin{proposition}[Rains, \cite{746807}]\label{thm:rainsuniqueness}
    The MDS code with parameters $((6,1,4))$ is unique up to local unitary equivalence. The same thing is true for the MDS codes with parameters $((5,2,3))$ and $((4,4,2))$.
\end{proposition}

Although every normalized 6-qubit AME state belongs to a unique $\U_2^{\times 6}$-orbit, there are many known ways to construct a representative in this orbit \cites{Borras_2007,PhysRevA.97.062326,Helwig:2013qoq,782103,MHein,4qubit}. In this paper, we use the 6-qubit AME state
\begin{equation*}\label{eq:AME}
\ket{\Psi}:= 2e^{-\im \pi/4}(T^\dagger)^{\otimes 3}(\ket{000000}+\ket{010101}+\ket{101010}+\ket{111111}),
\end{equation*}
where $T$ is the $4\times 4$ unitary matrix
\begin{equation}\label{eq:T}
    T:=\frac{1}{\sqrt{2}}\begin{pmatrix}
        1 & 0 & 0 & 1\\
        0 & \im & \im & 0\\
        0 & -1 & 1 & 0\\
        \im & 0 & 0 & -\im
    \end{pmatrix}.
\end{equation}
By \Cref{prop:rains}, the image of $\Tr_{\{1\}}(\ket{\Psi}\bra{\Psi})$ is a $((5,2,3))$ code $\mathcal{C}_1$. Tensor contraction gives us a natural basis of $\mathcal{C}_1$ consisting of the vectors $\ket{\psi_i}=\bra{i}_1\ket{\Psi}$, as we can see from the calculation
\[
\Tr_{\{1\}}(\ket{\Psi}\bra{\Psi})=\Tr_{\{1\}}\Big(\sum_{i,j=0}^1\ket{i}\bra{j}\otimes\ket{\psi_i}\bra{\psi_j}\Big)=\sum_{i=0}^1 \ket{\psi_i}\bra{\psi_i}.
\]
To simplify notation, define the vectors
\begin{align}
\begin{split}
\ket{u_1} &:= \ket{0000}+\ket{1111}, \quad
\ket{u_2} := \ket{0011}+\ket{1100}, \\
\ket{u_3} &:= \ket{0101}+\ket{1010}, \quad
\ket{u_4} := \ket{0110}+\ket{1001}.
\end{split}
\end{align}
Then the basis elements $\ket{\psi_i}$ that span $\mathcal{C}_1$ can be written as
\begin{align*}
\ket{\psi_0} &:= \ket{0}(\ket{u_1}-\im\ket{u_2})+\ket{1}(\im\ket{u_3}+\ket{u_4}),\\
\ket{\psi_1} &:= \ket{0}(\ket{u_3}+\im\ket{u_4})+\ket{1}(-\im\ket{u_1}+\ket{u_2}).
\end{align*}
Similarly, the image of $\Tr_{\{1,2\}}(\ket{\Psi}\bra{\Psi})$ is a $((4,4,2))$ code $\mathcal{C}_2$. Tensor contraction gives us a basis of $\mathcal{C}_2$ consisting of the vectors $\ket{\psi_{ij}}=\bra{j}_1\ket{\psi_i}$, which can be written as
\begin{align*}
\ket{\psi_{00}} &:= \ket{u_1}-\im\ket{u_2}, \quad
\ket{\psi_{01}} := \im\ket{u_3}+\ket{u_4}, \\
\ket{\psi_{10}} &:= \ket{u_3}+\im\ket{u_4}, \quad
\ket{\psi_{11}} := -\im\ket{u_1}+\ket{u_2}.
\end{align*}
From the above, it is clear that $\{\ket{u_i}\}$ is also a basis of $\mathcal{C}_2$. It might be helpful to note that these codes are named so that $\mathcal{C}_i$ is $2^i$-dimensional.

\subsection{A unitary change of coordinates}\label{sec:unitarychange} The $4\times 4$ unitary matrix $T$ defined in \eqref{eq:T} has an important property: it is used to define the double cover
\[
\pi:\SL_2\times\SL_2\to\SO_4,\quad \pi(A,B)= T(A\otimes B)T^\dagger.
\]
The space $\CC^2\otimes \CC^2$ is naturally identified with $\CC^4$ by the correspondence of basis vectors $\ket{i}\ket{j}\mapsto \ket{2i+j}$ for $i,j\in\{ 0,1\}$. Consider the unitary map
\begin{equation}\label{eq:changec}
(\CC^2)^{\otimes 2m}\to (\CC^4)^{\otimes m},\quad\ket{\varphi}\mapsto T^{\otimes m}\ket{\varphi}.
\end{equation}
Under this change of coordinates,
\[
A_1\otimes B_1\otimes\dots\otimes A_m\otimes B_m\in \SL_2^{\otimes 2m}\mapsto T(A_1\otimes B_1)T^\dagger\otimes\dots\otimes T(A_m\otimes B_m)T^\dagger\in\SO_4^{\otimes m}.
\]
Therefore, the $\SL_2^{\times 2m}$-module $(\CC^2)^{\otimes 2m}$ is equivalent to the $\SO_4^{\times m}$-module $(\CC^4)^{\otimes m}$. Furthermore, we have
\begin{equation}\label{eq:unitary-unitary}
\pi(\SU_2\times \SU_2)=\SO_4(\CC)\cap\U_4=\SO_4(\RR).
\end{equation}
So the $\SU_2^{\times 2m}$-module $(\CC^2)^{\otimes 2m}$ is equivalent to the $\SO_4(\RR)^{\times m}$-module $(\CC^4)^{\otimes m}$ via \eqref{eq:changec}.

In previous work, the map \eqref{eq:changec} was used to develop an algorithm for finding unique normal forms of $n$-qubit states under the action of $\SL_2^{\times n}$ \cite{tensordecompositions}. We shall see that in the 4-qubit case, \eqref{eq:changec} gives us the embedding $(\CC^2)^{\otimes 4}\to\mathfrak{so}_8$.

\subsection{Construction of 3-uniform states}\label{sec:3uniform} The 6-qubit AME state $\ket{\Psi}$ introduced in \Cref{sec:main} is a member of an infinite family of $3$-uniform states. The general construction works as follows. We take the $m$-ququart generalized GHZ state \[\ket{GHZ(m)}:=\sum_{i=0}^3 \ket{i}^{\otimes m}\in (\CC^4)^{\otimes m}\]
then apply the inverse of \eqref{eq:changec}. The result is a 3-uniform $2m$-qubit state.
\begin{theorem}\label{thm:3uniform}
    The state $\ket{\Psi_m}=(T^\dagger)^{\otimes m}\ket{GHZ(m)}\in (\CC^2)^{\otimes 2m}$ is 3-uniform for all $m\geq 3$.
\end{theorem}
\begin{proof}
    A direct check verifies the claim for $m=3$. Suppose $m>3$. Let $V_i\cong\CC^2$ for $1\leq i\leq 2m$ and let $W_j\cong\CC^4$ for $1\leq j\leq m$. We apply the convention that $V_{2j-1}\otimes V_{2j} = W_j$ for each $j$. For any 3-element subset $\{i_1,i_2,i_3\}\subset\{1,\dots,2m\}$ we can choose a subset $\{j_1,j_2,j_3\}\subset\{1,\dots,m\}$ such that $V_{i_1}\otimes V_{i_2}\otimes V_{i_3}\subset W_{j_1}\otimes W_{j_2}\otimes W_{j_3}$. Tracing out the complement $W$ of $W_{j_1}\otimes W_{j_2}\otimes W_{j_3}$ (that is, $W$ is the tensor product of copies of $\CC^4$ corresponding to the indices in the complement of $\{j_1,j_2,j_3\})$, we have
    \begin{align*}
    \Tr_{W}(\ket{\Psi_m}\bra{\Psi_m})&=\Tr_{W}((T^\dagger)^{\otimes m}\ket{GHZ(m)}\bra{GHZ(m)}T^{\otimes m}) \\
    &=(T^\dagger)^{\otimes 3}\sum_{i=0}^3 \ket{iii}\bra{iii}T^{\otimes 3}.
    \end{align*}
    To show that the reduction of $\ket{\Psi_m}\bra{\Psi_m}$ to $V_{i_1}\otimes V_{i_2}\otimes V_{i_3}$ is proportional to the identity, we check that
    \[
    \Tr_V\Big((T^\dagger)^{\otimes 3}\sum_{i=0}^3 \ket{iii}\bra{iii}T^{\otimes 3}\Big)\propto I
    \]
    for any $3$-qubit subsystem $V\subset (\CC^2)^{\otimes 6}$.
\end{proof}

\subsection{Local symmetries and transversal gates}\label{sec:stab} Let $\mathcal{C}$ be a subspace of $(\CC^D)^{\otimes n}$. We denote the stabilizer of $\mathcal{C}$ in $\SL_D^{\otimes n}$ with
\[
N(\mathcal{C}):=\{g\in \SL_D^{\otimes n}:g\mathcal{C}=\mathcal{C}\}.
\]
The rationale for our choice of notation should become clear in \Cref{sec:vinberg}. There exists a natural representation $\mu:N(\mathcal{C})\to \GL(\mathcal{C})$ defined by $\mu(g)\ket{\varphi}=g\ket{\varphi}$ for all $g\in\SL_D^{\otimes n}$ and $\ket{\varphi}\in\mathcal{C}$. A choice of basis for $\mathcal{C}$ gives coordinates to the image of $\mu$. That is, if $\{\ket{\varphi_i}\}_{i=1}^K$ is a basis of $\mathcal{C}$, then $\mu(g)$ is a $K\times K$ matrix with entries $\mu(g)_{ji}$ given by
\[
g\ket{\varphi_i}=\sum_j \mu(g)_{ji}\ket{\varphi_j}.
\]
We define the group $W(\mathcal{C}):=\mu(N(\mathcal{C}))$ to be the image of $\mu$.

Viewing $\mathcal{C}$ as a quantum error correcting code, a \textit{transversal gate} of $\mathcal{C}$ is an element $g\in N(\mathcal{C})$ such that $g\in\SU_D^{\otimes n}$. The transversal gates of a code are of interest because they can be used to apply operations on encoded states that do not propagate local errors \cite{PhysRevLett.102.110502}.

\begin{proposition}\label{prop:stab}
    Let $\ket{\varphi}\in (\CC^D)^{\otimes (n+1)}$ such that $\Tr_{\{1\}}(\ket{\varphi}\bra{\varphi})$ has rank $D$. Let $\mathcal{C}\subset(\CC^D)^{\otimes n}$ be the subspace spanned by
    \[\{\ket{\varphi_i}=\bra{i}_1\ket{\varphi}:0\leq i\leq D-1\}.\]
    This choice of basis gives coordinates to the image of the representation $\mu:N(\mathcal{C})\to\GL(\mathcal{C})$. Given $z\in\CC\setminus\{0\}$ and $g=g_0\otimes g_1\otimes\dots\otimes g_n$ with $g_i\in \SL_D$, the following are equivalent: 
    \begin{enumerate}
        \item[1.] $z g\ket{\varphi}=\ket{\varphi}$,
        \item[2.] $g_1\otimes\dots\otimes g_{n}\in N(\mathcal{C})$ and $\mu(g_1\otimes\dots\otimes g_{n})=(z g_0^\top)^{-1}$.
    \end{enumerate}
\end{proposition}
\begin{proof}
    First notice that $z g\ket{\varphi}=\ket{\varphi}$ if and only if \[(I\otimes g_1\otimes\dots\otimes g_{n})\ket{\varphi}=\frac{1}{z}(g_0^{-1}\otimes I\otimes\dots\otimes I)\ket{\varphi}.\] Since $\ket{\varphi}=\sum_{i=0}^{D-1} \ket{i}\ket{\varphi_i}$, the above expands to
    \begin{align*}
         \sum_{i=0}^{D-1}\ket{i}\otimes (g_1\otimes\dots\otimes g_{n})\ket{\varphi_i}&=\sum_{i=0}^{D-1}(z g_0)^{-1}\ket{i}\otimes\ket{\varphi_i} \\
         &= \sum_{i=0}^{D-1}\left(\sum_{j=0}^{D-1}(z g_0)^{-1}_{ji}\ket{j}\right)\otimes\ket{\varphi_i} \\
         &=\sum_{j=0}^{D-1}\ket{j}\otimes\left(\sum_{i=0}^{D-1} (z g_0)^{-1}_{ji}\ket{\varphi_i}\right) \\
         &=\sum_{i=0}^{D-1}\ket{i}\otimes\left(\sum_{j=0}^{D-1} (z g_0^\top)^{-1}_{ji}\ket{\varphi_j}\right),
    \end{align*}
    which is equivalent to
    $
(g_1,\dots,g_{n}).\ket{\varphi_i}=(zg_0^\top)^{-1}.\ket{\varphi_i}
    $
    for all $0\leq i\leq D-1$.
\end{proof}

Proposition~\ref{prop:stab} tells us that the group of \textit{local symmetries}
\[
S(\ket{\Psi}):=\{g\in \GL_2^{\otimes 6}:g\ket{\Psi}=\ket{\Psi}\}
\]
of $\ket{\Psi}$ and the group $N(\mathcal{C}_1)$ are, in a sense, the same object, as one determines the other. Furthermore, the groups $S(\ket{\Psi})$ and $N(\mathcal{C}_2)$ are closely related. To illustrate, in what follows we first calculate $N(\mathcal{C}_2)$ then apply Proposition~\ref{prop:stab} successively to find $S(\ket{\Psi})$ and $W(\mathcal{C}_1)$. We discover that all of these groups consist of local unitary operators, which implies that $N(\mathcal{C}_1)$ and $N(\mathcal{C}_2)$ are the groups of transversal gates of the codes $\mathcal{C}_1$ and $\mathcal{C}_2$, respectively. The caveat is that we need to switch back and forth between two choices of coordinates via the invertible map \eqref{eq:changec}.

\begin{proposition}\label{lem:lem1}
    Let $g_1,g_2\in\SO_4$. The following are equivalent:
    \begin{itemize}
        \item $g_1Dg_2^\top$ is diagonal whenever $D\in\CC^{4\times 4}$ is diagonal,
        \item $g_1=PD_1$ and $g_2=PD_2$ for some permutation matrix $P$ and diagonal matrices $D_1,D_2$ with diagonal entries in $\{\pm 1\}$.
    \end{itemize}
\end{proposition}
\begin{proof}
    Let $D=\text{diag}(\lambda_1,\dots,\lambda_4)$ such that the squared eigenvalues $\lambda^2_i$ are distinct. If $g_1Dg_2^\top$ is diagonal, so is $g_1D^2 g_1^\top$. Since conjugation preserves eigenvalues, we must have
    \[
    g_1 D^2 g_1^\top=\text{diag}(\lambda_{\sigma(1)}^2,\dots,\lambda_{\sigma(4)}^2)=P_\sigma D^2 P_\sigma^\top,
    \]
    where $P_\sigma$ is the permutation matrix corresponding to $\sigma\in\mathfrak{S}_4$. Then $P_\sigma^\top g_1$ commutes with $D^2$, so it preserves the eigenspaces of $D^2$. By choice of $D$, every eigenspace of $D^2$ is 1-dimensional. Then $P_\sigma^\top g_1$ is diagonal. We conclude that $g_1=P_\sigma D_1$ for some diagonal $D_1$. Since $P_{\sigma}\in\Or_4$ and $g_1\in \Or_4$, we have $D_1\in\Or_4$ and the diagonal entries of $D_1$ must be equal to $\pm 1$. The same argument applies to $g_2$. This proves one direction of the claim. The converse is clear.
\end{proof}

\begin{proposition}\label{prop:diag_code}
    The group $N(\mathcal{C}_2)$ is a subgroup of $\SU_2^{\otimes 4}$. Hence, $N(\mathcal{C}_2)$ is the set of transversal gates of $\mathcal{C}_2$.
\end{proposition}
\begin{proof}
    The $\SO_4^{\times 2}$-modules $\CC^{4\times 4}$ and $\CC^4\otimes\CC^4$ are equivalent by the invertible map $\CC^{4\times 4}\to \CC^4\otimes\CC^4$ defined by $\ket{i}\bra{j}\mapsto\ket{ij}$. The group action on $\CC^{4\times 4}$ is given by $(g_1,g_2).A=g_1Ag_2^\top$ for $g_1,g_2\in \SO_4$. Moreover, the $\SO_4^{\times 2}$-module $\CC^4\otimes\CC^4$ is equivalent to the $\SL_2^{\times 2}$-module $(\CC^2)^{\otimes 4}$ by the invertible map \eqref{eq:changec}. Therefore, the $\SO_4^{\times 2}$-module $\CC^{4\times 4}$ is equivalent to the $\SL_2^{\times 2}$-module $(\CC^2)^{\otimes 4}$. Likewise, the subspace $\mathcal{D}\subset \CC^{4\times 4}$ of diagonal matrices corresponds to the code $\mathcal{C}_2$. Indeed, a calculation gives
    \begin{align*}
    T^\dagger\otimes T^\dagger \ket{00}&=\frac{1}{2}(\ket{u_1}+\ket{u_2}), \\
    T^\dagger\otimes T^\dagger \ket{11}&=\frac{1}{2}(-\ket{u_3}-\ket{u_4}), \\
    T^\dagger\otimes T^\dagger \ket{22}&=\frac{1}{2}(\ket{u_3}-\ket{u_4}), \\
    T^\dagger\otimes T^\dagger \ket{33}&=\frac{1}{2}(-\ket{u_1}+\ket{u_2}).
    \end{align*}
    By \Cref{lem:lem1}, if $g\in\SO_4^{\times 2}$ such that $g.\mathcal{D}=\mathcal{D}$, then $g\in\SO_4(\RR)^{\times 2}$. Then, by \eqref{eq:unitary-unitary}, $N(\mathcal{C}_2)\leq \SU_2^{\otimes 4}$.
\end{proof}

\begin{theorem}[Local symmetries of $\ket{\Psi}$]\label{thm:GHZ3}
    Let $g\in\SO_4^{\otimes 3}$ and $z\in\CC\setminus\{0\}$. The following are equivalent:
    \begin{itemize}
        \item[1.] $zg\ket{GHZ(3)}=\ket{GHZ(3)}$,
        \item[2.] $zg=PD_1\otimes PD_2\otimes PD_3$ where $P\in \SO_4$ is a permutation matrix, each $D_i\in \SO_4$ is diagonal, and $D_1 D_2 D_3=I$.
    \end{itemize}
    If the above holds, then $z\in\{\pm 1\}$ and $zg\in \SO_4^{\otimes 3}$.
\end{theorem}
\begin{proof}
    Let $g=g_1\otimes g_2\otimes g_3$ with each $g_i\in \SO_4$ and suppose $zg\ket{GHZ(3)}=\ket{GHZ(3)}$. Let $\mathcal{C}\subset\CC^4\otimes\CC^4$ be the subspace spanned by the vectors $\bra{i}_1\ket{GHZ(3)}=\ket{ii}$ for $0\leq i\leq 3$. Then, by \Cref{prop:stab}, $g_2\otimes g_3\in N(\mathcal{C})$. By the equivalence of $\SO_4^{\times 2}$-modules $\CC^4\otimes \CC^4$ and $\CC^{4\times 4}$, \Cref{lem:lem1} implies that $g_2=PD_2$ and $g_3=PD_3$ for some permutation matrix $P$ and diagonal matrices $D_2,D_3$. By symmetry of $\ket{GHZ(3)}$, $g_1=PD_1$ for some diagonal $D_1$. Let $\lambda_i^j$ denote the $i$th diagonal entry of $D_j$ and let $\sigma\in\mathfrak{S}_4$ be the permutation corresponding to $P^{-1}$. Then $zg\ket{GHZ(3)}=\ket{GHZ(3)}$ implies that
    \begin{align*}
        \frac{1}{z}\sum_{i=0}^3 \ket{iii} &= g\sum_i \ket{iii} \\
        &= P^{\otimes 3}\sum_{i=0}^3 \lambda_i^{1}\lambda_i^{2}\lambda_i^{3}\ket{iii} \\
        &= \sum_{i=0}^3\lambda_{\sigma(i)}^{1}\lambda_{\sigma(i)}^{2}\lambda_{\sigma(i)}^{3}\ket{iii}.
    \end{align*}
    We conclude that $zD_1 D_2 D_3=I$ or $D_1' D'_2 D'_3=I$ with $D_1'=zD_1$, $D_2'=D_2$, and $D_3'=D_3$. Since each $\lambda_i^j\in\{\pm 1\}$, we have $z\in\{\pm1\}$. Since each $g_i\in \SO_4$, we have $1=\det(g_i)=\det(P)\det(D_i)$ or $\det(D_i')=\det(D_i)=1/\det(P)$. Then
    $
    \det^3(D'_i)=\det(D'_1 D'_2 D'_3)=\det(I)=1,
    $
    which implies that $\det(D'_i)=1$ and $D'_i\in \SO_4$. It follows that $P\in \SO_4$ also. This proves one direction of the claim. The converse is clear.
\end{proof}
Note that every $h\in \GL_2^{\otimes 6}$ has the form $h=zh'$ for some $h'\in \SL_2^{\otimes 6}$ and $z\in\CC\setminus\{0\}$. Since $\ket{\Psi}\propto(T^\dagger)^{\otimes 3}\ket{GHZ(3)},$ \Cref{thm:GHZ3} gives us the local symmetries $S(\ket{\Psi})$ after applying a change of coordinates. The last sentence of \Cref{thm:GHZ3} emphasizes that $S(\ket{\Psi})\leq \SL_2^{\otimes 6}$.

\begin{table}
\begin{center}
\begin{tabular}{ c | c }
\toprule
$U$ & $T^\dagger UT$ \\
\midrule\\
\addlinespace[-2ex]
$\begin{pmatrix}  0 & 0 & 1 & 0 \\ 1 & 0 & 0 & 0 \\ 0 & 1 & 0 & 0 \\ 0 & 0 & 0 & 1 \end{pmatrix}$ &
$\frac{1}{2}\begin{pmatrix}
1-\im & 1-\im \\
-1-\im & 1+\im
\end{pmatrix}
\otimes
\frac{1}{2}\begin{pmatrix}
    1+\im & -1-\im \\
    1-\im & 1-\im
\end{pmatrix}$\\

\addlinespace[0.75ex]
\midrule
\addlinespace[0.75ex]

$\begin{pmatrix}      0 & 1 & 0 & 0 \\
    1 & 0 & 0 & 0 \\
    0 & 0 & 0 & 1 \\
    0 & 0 & 1 & 0 \end{pmatrix}$ &
$\begin{pmatrix}
    \im & 0\\
    0 & -\im
\end{pmatrix}
\otimes
\begin{pmatrix}
    0 & 1\\
    -1 & 0
\end{pmatrix}$\\

\addlinespace[0.75ex]
\midrule
\addlinespace[0.75ex]

$\begin{pmatrix}  1 & 0 & 0 & 0 \\
    0 & 1 & 0 & 0 \\
    0 & 0 & -1 & 0 \\
    0 & 0 & 0 & -1 \end{pmatrix}$ &
$\begin{pmatrix}
    0 & \im\\
    \im & 0
\end{pmatrix}
\otimes
\begin{pmatrix}
    0 & \im\\
    \im & 0
\end{pmatrix}$\\
\bottomrule
\end{tabular}
\end{center}
\bigskip
\caption{Generators of the subgroup of elements in $W(\mathfrak{so}_8)$ with positive determinant as tensor products of matrices in $\SU_2$ under the change of coordinates associated to $T^\dagger$.}\label{table1}
\end{table}

\begin{proposition}\label{prop:trans}
    The group $N(\mathcal{C}_1)$ is a subgroup of $\SU_2^{\otimes 5}$. Hence, $N(\mathcal{C}_1)$ is the set of transversal gates of $\mathcal{C}_1$.
\end{proposition}
\begin{proof}
    Let $g_1\otimes\dots\otimes g_5\in N(\mathcal{C}_1)$ with $g_i\in \SL_2$. By \Cref{prop:stab}, $g\ket{\Psi}=\ket{\Psi}$, where $g=g_0\otimes g_1\otimes\dots\otimes g_5$ and $\mu(g_1\otimes\dots\otimes g_5)=(g_0^{-1})^\top$. Since $T^{\otimes 3}\ket{\Psi}$ is a scalar multiple of $\ket{GHZ(3)}$, we can apply \Cref{thm:GHZ3} to conclude that $T(g_i\otimes g_{i+1})T^\dagger\in \SO_4(\RR)$ for each $i\in\{0,2,4\}$. By \eqref{eq:unitary-unitary}, we have $g_i\in \SU_2$ for all $i$.
\end{proof}

\begin{proposition}\label{thm:grp}
    The group $W(\mathcal{C}_1)$ is generated by $\im X, \im Z,$ and $H$, with $H$ defined as in \Cref{thm:mainresult}.
\end{proposition}
\begin{proof}
    Let $g\in \SL_2^{\otimes 6}$ and $z\in\CC\setminus\{0\}$ such that $zg\ket{\Psi}=\ket{\Psi}$.
    By \Cref{thm:GHZ3}, $zg=g_1\otimes g_2\otimes\dots\otimes g_6$ such that $g_i\in \SL_2$ and $T(g_1\otimes g_2)T^\dagger=PD$ for some permutation matrix $P\in \SO_4$ and diagonal $D\in\SO_4$. The set of all $PD$ of this form is a group generated by the three matrices in the left column of \Cref{table1}; although not relevant to the proof, this group consists of all elements $w\in W(\mathfrak{so}_8)$ in the Weyl group of $\mathfrak{so}_8$ such that $\det(w)>0$. The right column of \Cref{table1} shows, for each generator $U$, how $T^\dagger U T$ decomposes as a Kronecker product of matrices in $\SU_2$. Applying \Cref{prop:stab}, we conclude that $W(\mathcal{C}_1)$ is generated by
    \[
    \im X=\begin{pmatrix}
        0 & \im \\
        \im & 0
    \end{pmatrix}, \quad \im Z=\begin{pmatrix}
        \im & 0 \\
        0 & -\im
    \end{pmatrix}, \quad \text{and}\quad
    H=\frac{1}{2}\begin{pmatrix}
        1+\im & 1+\im \\
        -1+\im & 1-\im
    \end{pmatrix}.
    \]
    These generators come from taking the inverse transpose of the left factors of the matrices that appear in the right column of \Cref{table1}.
\end{proof}

\subsection{The binary tetrahedral group}\label{sec:24} The generators of $W(\mathcal{C}_1)$ given by \Cref{thm:grp} are elements of $\SU_2$. Applying the standard double cover $\SU_2\to \SO_3(\RR)$, the generators map to the following:
\[
\im X \mapsto \begin{pmatrix}
    1 & 0 & 0 \\
    0 & -1 & 0 \\
    0 & 0 & -1
\end{pmatrix},\quad 
\im Z \mapsto \begin{pmatrix}
    -1 & 0 & 0 \\
    0 & -1 & 0 \\
    0 & 0 & 1
\end{pmatrix},\quad 
H \mapsto \begin{pmatrix}
    0 & 0 & 1 \\
    1 & 0 & 0\\
    0 & 1 & 0
\end{pmatrix}
\]
These $3\times 3$ matrices generate the tetrahedral group, which represents the rotational symmetries of the tetrahedron with vertices $(1,1,1)$, $(1,-1,-1)$, $(-1,1,-1)$, and $(-1,-1,1)$. The tetrahedral group has order 12; it acts as even permutations on the set of vertices and is isomorphic to the alternating group on 4 elements. Since $W(\mathcal{C}_1)$ contains $-I$, it is the double cover of the tetrahedral group, a group of order 24 known as the \textit{binary tetrahedral group}. Thus, by \Cref{prop:trans}, we recover the known fact that the transversal gates of $\mathcal{C}_1$ implement the binary tetrahedral group \cite{PhysRevLett.133.240603}.

\section{Vinberg's theory and some consequences}\label{sec:vinberg}

\subsection{A graded Lie algebra}\label{sec:graded}
Let $\mathfrak{g}:=\mathfrak{so}_{8}$ be the complex Lie algebra consisting of $8\times 8$ skew-symmetric matrices with the Lie bracket given by the commutator $[M,N]=MN-NM$ for $M,N\in\mathfrak{g}$. Let $G=\SO_{8}$ be the corresponding connected Lie group. The adjoint representation $\text{ad}:\mathfrak{g}\to\text{End}(\mathfrak{g})$ of the Lie algebra is defined by the action $M.N=[M,N]$. The adjoint representation $\text{Ad}:G\to \GL(\mathfrak{g})$ of the Lie group is defined by conjugation: $g.M=gMg^\top$, where $g\in G$. Up to a scalar multiple, the Killing form $\mathfrak{g}\times\mathfrak{g}\to\CC$ is given by $(M,N)\mapsto\Tr(MN)$. Note the invariance of the Killing form:
\begin{equation}\label{eq:killing}
\text{tr}([K,M]N)=\text{tr}(M[N,K]),\quad\text{for all $K,M,N\in\mathfrak{g}$.}
\end{equation}
Each element of $\mathfrak{g}$ can be viewed as a $2\times 2$ block matrix with $4\times 4$ blocks. From this we see that the space decomposes as $\mathfrak{g}=\mathfrak{g}_0\oplus\mathfrak{g}_1$, where
\begin{align*}
\mathfrak{g}_0:=\left\{\begin{pmatrix}
    E_1 & 0 \\
    0 & E_2
\end{pmatrix}: E_1,E_2\in\mathfrak{so}_4\right\}\quad\text{and}\quad
\mathfrak{g}_1 := \left\{\begin{pmatrix}
    0 & A \\
    -A^\top & 0
\end{pmatrix}: A\in\CC^{4\times 4} \right\}.
\end{align*}
This is in fact a $\mathbb{Z}_2$-grading of the Lie algebra:
$
[\mathfrak{g}_i,\mathfrak{g}_j]\subset\mathfrak{g}_{i+j \text{ (mod 2)}}.
$
In particular, we have $[\mathfrak{g}_0,\mathfrak{g}_1]\subset\mathfrak{g}_1$ so that the adjoint action induces a representation
\begin{equation*}\label{eq:ad}
\mathfrak{g}_0\to\text{End}(\mathfrak{g}_1),\quad E\mapsto \text{ad}(E)|_{\mathfrak{g}_1}.
\end{equation*}
Let $G_0:=\exp(\mathfrak{g}_0)\cong \SO_4^{\times 2}$ be the subgroup of $G$ corresponding to $\mathfrak{g}_0$. The Lie algebra representation above corresponds to the Lie group representation
\begin{equation}
\xi:G_0\to \GL(\mathfrak{g}_1),\quad g\mapsto \text{Ad}(g)|_{\mathfrak{g}_1}.
\end{equation}
Let $L:=\xi(G_0)$ be the image of this representation. We are interested in the grade-1 component $\mathfrak{g}_1$ as an $L$-module. Consider the invertible map
\[
\CC^{4\times 4}\to\mathfrak{g}_1,\quad
A\mapsto M_A :=\begin{pmatrix}
    0 & A \\
    -A^\top & 0
\end{pmatrix}.
\]
Given $A\in\CC^{4\times 4}$ and $h_1,h_2\in\SO_4$, we have
\[
\text{Ad}\begin{pmatrix}
    h_1 & 0 \\ 0 & h_2
\end{pmatrix}M_A=M_{h_1Ah_2^\top}.
\]
Thus, the corresponding $\SO_4^{\times 2}$-action on $\CC^{4\times 4}$ is given by $(h_1,h_2).A=h_1 Ah_2^\top$ for $h_1,h_2\in\SO_4$. Combining this observation with the discussion from \Cref{sec:unitarychange}, we see that the following are equivalent:
\begin{itemize}
    \item[1.] the $L$-module $\mathfrak{g}_1$,
    \item[2.] the $\SO_4^{\times 2}$-module $\CC^{4\times 4}$,
    \item[3.] the $\SO_4^{\times 2}$-module $\CC^4\otimes\CC^4$,
    \item[4.] the $\SL_2^{\times 4}$-module $(\CC^2)^{\otimes 4}$.
\end{itemize}
We get an embedding $(\CC^2)^{\otimes 4}\to \mathfrak{g}$ by composing the invertible maps between consecutive items of this list.

\subsection{Cartan subspaces}\label{sec:cartansubspaces}
Recall that an element $M\in\mathfrak{g}$ is \textit{semisimple} if $\text{ad}(M)$ is diagonalizable. Let $\mathfrak{a}$ be a subspace of $\mathfrak{g}_1$ with the following properties:
\begin{itemize}
    \item[1.] $\mathfrak{a}$ is \textit{abelian}, that is, $[\mathfrak{a},\mathfrak{a}]=0$, and
    \item[2.] every $M\in\mathfrak{a}$ is semisimple.
\end{itemize}
If $\mathfrak{a}$ is maximal among all subspaces of $\mathfrak{g}_1$ with these properties, then $\mathfrak{a}$ is a \textit{Cartan subspace} of $\mathfrak{g}_1$.
Any two Cartan subspaces are conjugate (\Cref{prop:conjugacy}), but our preferred one is
\[
\mathfrak{h}:=\{M_A:\text{$A\in\CC^{4\times 4}$ is diagonal}\}.
\]
Note that $\mathfrak{h}$ is also a Cartan subalgebra of $\mathfrak{g}$ in the usual Lie-theoretical sense. From the proof of Proposition~\ref{prop:diag_code}, we see that $\mathfrak{h}\subset\mathfrak{g}_1$ corresponds to the code $\mathcal{C}_2\subset (\CC^2)^{\otimes 4}$.

\begin{proposition}[{\cite[Proposition 3]{Vinberg_1976}}]\label{prop:semiclosed}
    An element $M\in\mathfrak{g}_1$ is semisimple if and only if its $L$-orbit is closed.
\end{proposition}

\begin{proposition}[{\cite[Theorem 1]{Vinberg_1976}}]\label{prop:conjugacy}
    If $\mathfrak{a}\subset\mathfrak{g}_1$ is a Cartan subspace, then there exists $g\in L$ such that $g\mathfrak{a}=\mathfrak{h}$.
\end{proposition}

If $\mathfrak{a}\subset\mathfrak{g}_1$ is an abelian subspace consisting of semisimple elements, then there exists $g\in L$ such that $g\mathfrak{a}\subset \mathfrak{h}$. Indeed, any such $\mathfrak{a}$ is contained in one that is maximal: a Cartan subspace. The claim then follows from \Cref{prop:conjugacy}. We end this section with a technical lemma about these subspaces of Cartan subspaces.

\begin{lemma}\label{lem:stab}
    Let $\mathfrak{a}$ be an abelian subspace of $\mathfrak{g}_1$ consisting of semisimple elements. Almost every $M\in\mathfrak{a}$ has the following property: if $g\in L$ such that $gM=M$, then $gN=N$ for all $N\in\mathfrak{a}$.
\end{lemma}
\begin{proof}
We may apply an element of $L$ and assume that $\mathfrak{a}\subset\mathfrak{h}$, by \Cref{prop:conjugacy}. The subspace $\mathfrak{a}$ corresponds to a subspace $\mathcal{D}\subset\CC^{4\times 4}$ consisting of diagonal matrices via the equivalence $\mathfrak{g}_1\cong\CC^{4\times 4}$ of $\SO_4^{\times 2}$-modules. We will prove the corresponding statement in $\CC^{4\times 4}$.
We may also replace $\mathcal{D}$ with $h.\mathcal{D}$, where $h\in \Or_4^{\times 2}$. To see this, suppose that $h.A\in h.\mathcal{D}$ has the property that $gh.A=h.A$ implies $gh.\mathcal{D}=h.\mathcal{D}$ for $g\in\SO_4^{\times 2}$. Then $h^{-1}gh.A=A$ implies $h^{-1}gh.\mathcal{D}=\mathcal{D}$. Since $\SO_4^{\times 2}$ is a normal subgroup of $\Or_4^{\times 2}$, the map $g\mapsto h^{-1}gh$ is an automorphism of $\SO_4^{\times 2}$.

Now, let $\{A^i:1\leq i\leq K\}$ be a basis of $\mathcal{D}$, where $A^i=\text{diag}(\lambda^i_1,\lambda^i_2,\lambda^i_3,\lambda^i_4)$. A parameterization of $\mathcal{D}$ is given by
\[
A=\sum_{i=1}^K z_iA^{i}=\text{diag}\Big(\sum_{i=1}^K\lambda^{i}_1z_i,\dots,\sum_{i=1}^K\lambda^{i}_4z_i\Big),
\]
where the $z_i$ are indeterminates. An element of $\mathcal{D}$ is obtained by evaluating $A$ at some $z\in\CC^K$. Conjugating by a permutation matrix in $\Or_4$, we can permute entries and thus assume that $A$ has a block diagonal form
\begin{equation}\label{eq:diagparam}
A = f_1(z) I_{m_1\times m_1}\oplus\dots\oplus f_s(z) I_{m_s\times m_s}
\end{equation}
for some distinct $f_i\in \CC[z_1,\dots,z_K]$ and positive integers $m_i$ such that $\sum_i m_i=4$. We may also assume that $f_i^2\neq f_j^2$ or $f_i\neq -f_j$ whenever $i\neq j$. If not, apply an element of $\Or_4^{\times 2}$ to flip the signs of the appropriate diagonal entries. Let $B=A|_v$ for some $v\in\CC^K$ to be determined. Suppose $g_1,g_2\in \SO_4$ such that $g_1 B g_2^\top=B$. We complete the proof by showing that $g_1 Dg_2^\top=D$ for all $D\in\mathcal{D}$. To this end, note that $g_1 B^2 g_1^\top = B^2$ and
\begin{equation}\label{eq:diagparam2}
B^2 = f^2_1(v) I_{m_1\times m_1}\oplus\dots\oplus f^2_s(v) I_{m_s\times m_s}.
\end{equation}
Since $g_1$ commutes with $B^2$, $g_1$ preserves the eigenspaces of $B^2$. For almost all $v$, the values $f^2_1(v),\dots,f^2_s(v)$ are distinct, which implies that the eigenspaces of $B^2$ correspond to the blocks of \eqref{eq:diagparam2}. Therefore, $g_1\in \SO_{m_1}\oplus\dots\oplus\SO_{m_s}$ is block diagonal with the same shape as \eqref{eq:diagparam2}. A similar argument applies for $g_2$. Then $g_1 Bg_2^\top=B$ implies that $g_1=g_2$. Thus, $(g_1,g_2)$ fixes \eqref{eq:diagparam} as a matrix of polynomials. That is, $(g_1,g_2)$ fixes every $D\in\mathcal{D}$.
\end{proof}

\subsection{Weyl group}\label{sec:weyl} In Vinberg's theory, the subgroup
$
N(\mathfrak{h}):=\{g\in L: g\mathfrak{h}=\mathfrak{h}\}
$
is known as the \textit{normalizer} of the Cartan subspace $\mathfrak{h}$. There exists a natural representation $\mu:N(\mathfrak{h})\to \GL(\mathfrak{h})$ defined by $\mu(g)M=gM$ for all $g\in N(\mathfrak{h})$ and $M\in\mathfrak{h}$. The \textit{Weyl group} of $\mathfrak{h}$ is defined as $W(\mathfrak{h}):=\mu(N(\mathfrak{h}))$.

In the basis $\{M_{\ket{i}\bra{i}}\}$ of $\mathfrak{h}$, the elements of $W(\mathfrak{h})$ are precisely the matrices of the form $PB$, where $P\in \Or_4$ is a permutation matrix and $D\in \SO_4$ is a special orthogonal diagonal matrix (for example, see \cite[Exercise~3, p.~179]{WallachGIT}). This group is also the Weyl group $W(\mathfrak{so_8})$ of the Lie algebra $\mathfrak{so}_8$ in the usual sense of Lie theory. Combining this with \Cref{prop:diag_code}, we see that the transversal gates of $\mathcal{C}_2$ implement $W(\mathfrak{so_8})$. The elements of the basis of $\mathcal{C}_2$ corresponding to each $M_{\ket{i}\bra{i}}$ are written in the proof of Proposition~\ref{prop:diag_code}.

The following is analogous to the Chevalley restriction theorem:
\begin{proposition}[{\cite[Theorem 7]{Vinberg_1976}}]\label{thm:isomorphism}
    The restriction map $f\mapsto f|_{\mathfrak{h}}$ for $f\in \CC[\mathfrak{g}]$ induces an isomorphism of invariant algebras $\CC[\mathfrak{g}]^L\cong \CC[\mathfrak{h}]^{W(\mathfrak{h})}$.
\end{proposition}

\subsection{Critical points} We say that $M\in\mathfrak{g}_1$ is \textit{critical} if the corresponding state in $(\CC^2)^{\otimes 4}$ is critical. A subspace $\mathfrak{a}\subset\mathfrak{g}_1$ is \textit{critical} if every $M\in\mathfrak{a}$ is critical. By \Cref{prop:obs}, a subspace in $\mathfrak{g}_1$ is critical if and only if it corresponds to a nontrivial pure code in $(\CC^2)^{\otimes 4}$. In particular, the Cartan subspace $\mathfrak{h}$ is critical and corresponds to the $((4,4,2))$ code $\mathcal{C}_2\subset(\CC^2)^{\otimes 4}$.

Note that $M\in\mathfrak{g}_1$ is critical if and only if
\begin{equation}\label{eq:altcritical}
    \text{tr}([E,M]M^\dagger)=0,\quad \forall E\in \mathfrak{g}_0.
\end{equation}
This follows from translating the characterization of critical points given by \Cref{prop:equiv}, which says that $x\in (\CC^2)^{\otimes 4}$ is critical if and only if $\langle Ex,x\rangle=0$ for all $E\in \text{Lie}(\SL_2^{\otimes 4})$, where angle brackets indicate the Hermitian inner product.
Applying the change of coordinates $x'=(T\otimes T)x$, the condition becomes $\langle Ex',x'\rangle=0$ for all $E\in\text{Lie}(\SO_4^{\otimes 2})$. Every $E\in\text{Lie}(\SO_4^{\otimes 2})$ has the form $E=E_1\otimes I+I\otimes E_2$ for some $E_1,E_2\in \mathfrak{so}_4$. Note that $\text{tr}(M_A M^\dagger_B)=2\text{tr}(AB^\dagger)$ for any $A,B\in \CC^{4\times 4}$. Therefore, $M_A\in\mathfrak{g}_1$ is critical if and only if
\[
\text{tr}([\begin{pmatrix}
    E_1 & 0 \\ 0 & E_2
\end{pmatrix},M_A]M_{A}^\dagger)=2\text{tr}((E_1A+AE_2^\top)A^\dagger)=0
\]
for every $E_1,E_2\in \mathfrak{so}_4$, as desired.

\begin{lemma}\label{lem:crit_semi}
    Every critical point in $\mathfrak{g}_1$ is semisimple.
\end{lemma}
\begin{proof}
    This follows from  \Cref{prop:semiclosed} and the third item of \Cref{prop:K-N}.
\end{proof}

\begin{theorem}\label{thm:git1}
    A state $\ket{\varphi}\in (\CC^2)^{\otimes 4}$ is critical if and only if it is $\SU_2^{\times 4}$-equivalent to a point in $\mathcal{C}_2$.
\end{theorem}
\begin{proof}
    By \Cref{prop:obs}, every point of $\mathcal{C}_2$ is critical. Hence, $g\ket{\varphi}$ is critical whenever $g\in \SU_2^{\otimes 4}$ and $\ket{\varphi}\in\mathcal{C}_2$. Conversely, suppose $\ket{\varphi}\in(\CC^2)^{\otimes 4}$ is critical. By \Cref{lem:crit_semi}, $\ket{\varphi}$ is semisimple as a point in $\mathfrak{g}_1$. Thus, by \Cref{prop:conjugacy}, $\ket{\varphi}$ is $\SL_2^{\times 4}$-equivalent to a point $\ket{\varphi'}\in\mathcal{C}_2$. Since $\ket{\varphi'}$ is critical, $\ket{\varphi}$ is $\SU_2^{\times 4}$-equivalent to $\ket{\varphi'}$ by \Cref{prop:K-N}.
\end{proof}
\begin{theorem}\label{thm:git2}
    Two states $\ket{\varphi},\ket{\psi}\in\mathcal{C}_2$ are $\SL_2^{\times 4}$-equivalent if and only if they are $W(\mathcal{C}_2)$-equivalent.
\end{theorem}
\begin{proof}
    It follows immediately from the definition of $W(\mathcal{C}_2)$ that if $\ket{\varphi}$ and $\ket{\psi}$ are $W(\mathcal{C}_2)$-equivalent, then $\ket{\varphi}$ and $\ket{\psi}$ are $\SL_2^{\times 4}$-equivalent. Conversely, suppose that $\ket{\varphi}$ and $\ket{\psi}$ are $\SL_2^{\times 4}$-equivalent. Then $f(\ket{\varphi})=f(\ket{\psi})$ for any $\SL_2^{\times 4}$-invariant $f\in\CC[(\CC^2)^{\otimes 4}]$. By \Cref{thm:isomorphism}, we have $f'(\ket{\varphi})=f'(\ket{\psi})$ for all $W(\mathcal{C}_2)$-invariant $f'\in \CC[\mathcal{C}_2]$. Since $W(\mathcal{C}_2)\cong W(\mathfrak{so}_8)$ is finite, this implies that $\ket{\varphi}$ and $\ket{\psi}$ are $W(\mathcal{C}_2)$-equivalent (see \Cref{thm:catquotient}).
\end{proof}
By \Cref{prop:K-N}, two critical states in $(\CC^2)^{\otimes 4}$ are $\SL_2^{\times 4}$-equivalent if and only if they are $\SU_2^{\times 4}$-equivalent. Thus, \Cref{thm:git1,thm:git2} imply \Cref{thm:wallach1u}.

\subsection{Note on the spin flip} Taking the negative conjugate transpose $M\mapsto -M^\dagger$ of a matrix $M\in\mathfrak{g}_1$ corresponds to an involutory $\RR$-linear map $(\CC^2)^{\otimes 4}\to (\CC^2)^{\otimes 4}$ known as the \textit{spin flip} \cite{PhysRevLett.80.2245}. For an arbitrary number $n$ of qubits, this map is defined by
\[
\ell:(\CC^2)^{\otimes n}\to (\CC^2)^{\otimes n},\quad\ell\ket{\varphi}=Y^{\otimes n}\ket{\bar{\varphi}}.
\]
Recall that a bar over a vector indicates the complex conjugate. The spin flip is the standard time reversal operation for spin $\frac{1}{2}$ particles \cite{sakurai_modern_2021}. Note the Cayley-Klein parameterization of the special unitary group:
\begin{equation}\label{eq:cayley-klein}
\SU_2=\left\{\begin{pmatrix}
    a & b\\
    -\bar{b} & \bar{a}
\end{pmatrix}:a,b\in\CC\:\land\:|a|^2+|b|^2=1\right\}.
\end{equation}
 From this it is simple to check that $U_i Y=Y\bar{U}_i$ for any $U_i\in\SU_2$. If $U=U_1\otimes\dots\otimes U_n$ where each $U_i\in\SU_2$, then
 \[U(\ell \ket{\varphi})=UY^{\otimes n}\ket{\bar{\varphi}}=Y^{\otimes n}\bar{U}\ket{\bar{\varphi}}=\ell (U\ket{\varphi}),\quad\forall\ket{\varphi}\in(\CC^2)^{\otimes n}.\]
 That is, $\ell$ is equivariant: it commutes with any $U\in \SU_2^{\otimes n}$. For $n=4$ qubits, this corresponds to the fact that $gM^\dagger g^\dagger=(gM g^\dagger)^\dagger$ for $M\in\mathfrak{g}_1$ and $g\in G_0\cap\U_8$.

\subsection{Four-qubit codes of distance two}\label{sec:1uniform} We present a sequence of lemmas leading to \Cref{cor:1uniform}, which tells us that every nontrivial 4-qubit pure code is equivalent to a subspace of $\mathcal{C}_2$. This generalizes Rains' result on the uniqueness of the $((4,4,2))$ code (\Cref{thm:rainsuniqueness}).

\begin{lemma}[Wallach, {\cite[Lemma 3.73]{WallachGIT}}]\label{lem:wallach}
The point $M\in \mathfrak{g}_1$ is critical if and only if $[M,M^\dagger]=0$.
\end{lemma}
\begin{proof}
    By \eqref{eq:killing}, we have
    $\text{tr} ([E,M]M^\dagger)=\text{tr}(E[M,M^\dagger])
    $
    for all $E\in \mathfrak{g}_0$.
    The restriction of the Killing form $\mathfrak{g}_0\times\mathfrak{g}_0\to\CC$ is nondegenerate. Therefore, $\text{tr} ([E,M]M^\dagger)=0$ for all $E\in\mathfrak{g}_0$ if and only if $[M,M^\dagger]=0.$ Using the condition \eqref{eq:altcritical} for being critical, we are done.
\end{proof}
\begin{lemma}\label{lem:critsubspace}
A subspace $\mathfrak{a}\subset \mathfrak{g}_1$ is critical if and only if $[\mathfrak{a},\mathfrak{a}^\dagger]=0$.
\end{lemma}
\begin{proof}
    It is clear from \Cref{lem:wallach} that if $[\mathfrak{a},\mathfrak{a}^\dagger]=0$, then $\mathfrak{a}$ is critical. Conversely, suppose $\mathfrak{a}$ is critical. Then for any $M,N\in\mathfrak{g}_1$ and $z\in \CC$ we have
    \[
    0=[M+z N,M^\dagger+\bar{z} N^\dagger]=\bar{z}[M,N^\dagger]+z[N,M^\dagger]=\bar{z}[M,N^\dagger]+(\bar{z}[M,N^\dagger])^\dagger.
    \]
    Setting $E_z=\bar{z}[M,N^\dagger]$, we have $E_z=-E_z^\dagger$. Since $E_z$ is skew-symmetric, $E_z=(E_z^{\dagger})^\top$ so $E_z$ is purely real. Since $z\in \CC$ is arbitrary, we must have $[M,N^\dagger]=0$.
\end{proof}
\begin{lemma}\label{lem:uniquecritcartan}
    Let $\mathfrak{a}\subset\mathfrak{g}_1$ be a critical abelian subspace. There exists $g\in G_0\cap \U_8$ such that $g.\mathfrak{a}\subset\mathfrak{h}$.
\end{lemma}
\begin{proof}
    by \Cref{lem:crit_semi}, $\mathfrak{a}$ consists of semisimple elements. By \Cref{prop:conjugacy}, there exists $g\in G_0$ such that $g.\mathfrak{a}\subset\mathfrak{h}$. It is left to show that $g$ can be chosen to be unitary. We can choose, by \Cref{lem:stab}, an element $M\in g.\mathfrak{a}$ such that $h.(g.\mathfrak{a})=g.\mathfrak{a}$ whenever $h\in G_0$ such that $h.M=M$. Since $g^{-1}.x\in\mathfrak{a}$ and $M\in\mathfrak{h}$ are both critical, by \Cref{prop:K-N}, there exists $\tilde{g}\in G_0\cap \U_8$ such that $\tilde{g}.(g^{-1}.M)=M$. This element has the form $\tilde{g}=hg$, where $h\in G_0$ such that $h.M=M$. It follows that $\tilde{g}.\mathfrak{a}=g.\mathfrak{a}$.
\end{proof}
\begin{corollary}\label{cor:wal}
    If $M\in\mathfrak{g}_1$ is critical, then there exists $g\in G_0\cap \U_8$ such that $g.M\in\mathfrak{h}$.
\end{corollary}
\begin{proof}
    If $M$ is critical, then $M$ spans a critical abelian subspace. The result follows from \Cref{lem:uniquecritcartan}.
\end{proof}
For the proof of \Cref{lem:lll}, let us work out what it means for elements of $\mathfrak{g}_1$ to commute in terms of their $4\times 4$ submatrices. For $A,B\in\CC^{4\times 4}$, we have
\[
[M_A,M_B]=\begin{pmatrix}
    -AB^\top+BA^\top & 0 \\
    0 & -A^\top B + B^\top A
\end{pmatrix}.
\]
Therefore, $[M_A,M_B]=0$ if and only if $AB^\top=BA^\top=(AB^\top)^\top$ and $B^\top A=A^\top B=(B^\top A)^\top$ are symmetric.
\begin{lemma}\label{lem:lll}
    Suppose $M,N\in \mathfrak{g}_1$ such that $N$ is critical. We have $[M,N]=0$ if and only if $[M,N^\dagger]=0$.
\end{lemma}
\begin{proof}
    It suffices to prove one direction because $N^\dagger$ is critical if $N$ is critical and $(N^\dagger)^\dagger=N$. Furthermore, we may assume that $N\in\mathfrak{h}$ since there exists $g\in G_0\cap \U_8$ such that $g.N\in\mathfrak{h}$ (\Cref{cor:wal}) and $N\mapsto N^\dagger$ commutes with the action of $g$.
    
    Let $A,B\in \CC^{4\times 4}$ such that $B$ is diagonal. If $[M_A,M_B]=0$, then $AB$ and $BA$ are symmetric. That is, for any $1\leq i<j\leq 4$ we have
    \begin{align*}
    0&=(AB)_{ij}-(AB)_{ji}=A_{ij}B_{jj}-A_{ji}B_{ii},\\
    0&=(BA)_{ij}-(BA)_{ji}=A_{ij}B_{ii}-A_{ji}B_{jj}.
    \end{align*}
    Adding and subtracting the equations above, we obtain the equivalent system:
    \begin{align*}
    0&=(A_{ij}-A_{ji})B_{ii}+(A_{ij}-A_{ji})B_{jj}=(A_{ij}-A_{ji})(B_{ii}+B_{jj}), \\
    0&=(A_{ij}+A_{ji})B_{jj}-(A_{ij}+A_{ji})B_{ii}=(A_{ij}+A_{ji})(B_{jj}-B_{ii}).
    \end{align*}
    But this holds if and only if
    \[
    0=(A_{ij}-A_{ji})(\bar{B}_{ii}+\bar{B}_{jj})\quad\text{and}\quad
    0=(A_{ij}+A_{ji})(\bar{B}_{jj}-\bar{B}_{ii}).\]
    Reasoning backwards, we conclude that $A\bar{B}$ and $\bar{B}A$ are symmetric. Thus, $[M_A,M_B^\dagger]=[M_A,-M_{\bar{B}}]=0$.
\end{proof}

\begin{theorem}\label{thm:general}
    Let $\mathfrak{a}\subset\mathfrak{g}_1$ be a critical subspace. There exists $g\in G_0\cap \U_8$ such that $g.\mathfrak{a}\subset\mathfrak{h}$.
\end{theorem}
\begin{proof}
    We show that $\mathfrak{a}$ is abelian; the result follows from \Cref{lem:uniquecritcartan}. Let $M,N\in\mathfrak{a}$. By \Cref{lem:critsubspace}, we have $[M,N^\dagger]=0$. By \Cref{lem:lll}, we have $[M,N]=0$.
\end{proof}

\begin{corollary}\label{cor:1uniform}
    Every nontrivial pure code $\mathcal{C}\subset(\CC^2)^{\otimes 4}$ is $\SU_2^{\times 4}$-equivalent to a subspace of $\mathcal{C}_2$.
\end{corollary}
\begin{proof}
    Since 4-qubit AME states do not exist \cite{HIGUCHI2000213}, there are no pure codes $\mathcal{C}\subset(\CC^2)^{\otimes 4}$ of distance greater than 2. Thus, the result follows from \Cref{thm:general}.
\end{proof}

\section{Main result}\label{sec:classification}

In this section, we prove our main result, \Cref{thm:mainresult}, which is split into \Cref{thm:first,thm:second,thm:separating}. We will need some basic results of algebraic geometry and invariant theory which can be found in Appendix~\ref{sec:algebraic_geometry} and Appendix~\ref{sec:invariant_theory}.

Let $V$ be a finite-dimensional real vector space. If $\mathcal{A}\subset V$ is a (affine) variety, $\textbf{I}(\mathcal{A})\subset \RR[V]$ is the ideal of $\mathcal{A}$. If $\mathcal{I}\subset \RR[V]$ is an ideal, $\textbf{V}_\RR(\mathcal{I})$ is the variety of $\mathcal{I}$.

\begin{table}
    \begin{tabular}{c}
    \toprule
        $x_{02}y_{01}-x_{01}y_{02}+x_{12}y_{11}-x_{11}y_{12}$ \\
        $x_{03}y_{00}-x_{00}y_{03}+x_{13}y_{10}-x_{10}y_{13}$ \\
        $x_{00}x_{03}+x_{10}x_{13}+y_{00}y_{03}+y_{10}y_{13}$\\
        $x_{01}x_{02}+x_{11}x_{12}+y_{01}y_{02}+y_{11}y_{12}$\\
        $x_{01}^2+x_{02}^2-x_{10}^2-x_{13}^2+y_{01}^2+y_{02}^2-y_{10}^2-y_{13}^2$\\
        $x_{00}^2+x_{03}^2-x_{11}^2-x_{12}^2+y_{00}^2+y_{03}^2-y_{11}^2-y_{12}^2$\\
        $x_{10}y_{00}+x_{11}y_{01}+x_{12}y_{02}+x_{13}y_{03}-x_{00}y_{10}-x_{01}y_{11}-x_{02}y_{12}-x_{03}y_{13}$ \\
        $x_{02}y_{00}-x_{03}y_{01}-x_{00}y_{02}+x_{01}y_{03}+x_{12}y_{10}-x_{13}y_{11}-x_{10}y_{12}+x_{11}y_{13}$\\
        $x_{01}y_{00}-x_{00}y_{01}-x_{03}y_{02}+x_{02}y_{03}+x_{11}y_{10}-x_{10}y_{11}-x_{13}y_{12}+x_{12}y_{13}$ \\
        $x_{00}x_{10}+x_{01}x_{11}+x_{02}x_{12}+x_{03}x_{13}+y_{00}y_{10}+y_{01}y_{11}+y_{02}y_{12}+y_{03}y_{13}$\\
        $x_{00}x_{02}-x_{01}x_{03}+x_{10}x_{12}-x_{11}x_{13}+y_{00}y_{02}-y_{01}y_{03}+y_{10}y_{12}-y_{11}y_{13}$\\
        $x_{00}x_{01}-x_{02}x_{03}+x_{10}x_{11}-x_{12}x_{13}+y_{00}y_{01}-y_{02}y_{03}+y_{10}y_{11}-y_{12}y_{13}$\\
        \bottomrule
    \end{tabular}
    \bigskip
    \caption{Generators of the ideal $\mathcal{J}$ from the proof of \Cref{thm:first}.}\label{table2}
\end{table}

\begin{theorem}\label{thm:first}
    Suppose $\ket{\psi}\in(\CC^2)^{\otimes 5}$ is AME. There exists $g\in\SU_2^{\times 5}$ such that $g.\ket{\psi}\in\mathcal{C}_1$.
\end{theorem}
\begin{proof}
    The $((5,2,3))$ code $\mathcal{C}_1$ is unique up to local unitary equivalence (\Cref{thm:rainsuniqueness}), so it suffices to show that every AME state is equivalent to a point of \textit{some} $((5,2,3))$ code. Since $\ket{\psi}$ is 2-uniform, the vectors $\bra{0}_1\ket{\psi}$ and $\bra{1}_1\ket{\psi}$ span a $((4,2,2))$ code (\Cref{prop:rains}). By \Cref{cor:1uniform}, we can apply a local unitary change of coordinates so that $\bra{0}_1\ket{\psi}$ and $\bra{1}_1\ket{\psi}$ are in $\mathcal{C}_2$. That is, we may assume that $\ket{\psi}\in \CC^2\otimes \mathcal{C}_2$. The proof proceeds as follows. Let $\mathcal{A}\subset \CC^2\otimes\mathcal{C}_2$ be the real variety of AME states in $\CC^2\otimes\mathcal{C}_2$. We show that this variety decomposes into irreducibles as $\mathcal{A}=\mathcal{A}_1\cup\mathcal{A}_2$, where $\mathcal{A}_1$ and $\mathcal{A}_2$ are orbits of $((5,2,3))$ codes under the action of a subgroup of $\SU_2^{\times 5}$.
    
    In what follows, we use, from \Cref{sec:main}, the bases $\{\ket{\psi_i}\}$ and $\{\ket{\psi_{ij}}\}$ of $\mathcal{C}_1$ and $\mathcal{C}_2$, respectively. These bases are related by the equation $\ket{\psi_{ij}}=\bra{j}_1\ket{\psi_i}$ so that
    \begin{align}\label{eq:bdef}
    \ket{\psi_0}=\ket{0}\ket{\psi_{00}}+\ket{1}\ket{\psi_{01}}\quad \text{and}\quad
    \ket{\psi_1}=\ket{0}\ket{\psi_{10}}+\ket{1}\ket{\psi_{11}}.
    \end{align}
    Every $\ket{\psi}\in\CC^2\otimes\mathcal{C}_2$ has the form
    \begin{align*}
\ket{\psi}=&\sum_{i=0}^1\ket{i}\otimes(z_{i0}\ket{\psi_{00}}+z_{i1}\ket{\psi_{01}}+z_{i2}\ket{\psi_{10}}+z_{i3}\ket{\psi_{11}})
    \end{align*}
    for some constants $z_{ij}\in\CC$. Let $z_{ij}=x_{ij}+\im y_{ij}$ where $x_{ij},y_{ij}\in \RR$. The variety $\mathcal{A}$ is cut out by the 12 quadratics listed in \Cref{table2}, which generate an ideal $\mathcal{J}\subset\RR[x_{ij},y_{ij}]$. These polynomials come from the fact that $\ket{\psi}$ is AME if and only if $\rho_S=\Tr_{S}(\ket{\psi}\bra{\psi})$ is proportional to the identity for every 3-element subset $S\subset\{1,\dots,5\}$. Thus, $\mathcal{A}$ is defined by the vanishing of
    $
    (\rho_S)_{ij}
    $ for $1\leq i<j\leq 4$ and $(\rho_S)_{11}-(\rho_S)_{ii}$ for $2\leq i\leq 4$. By taking real and imaginary parts, this collection of functions corresponds to a collection of real polynomials. After simplification, we obtain \Cref{table2}.
    
     Using the Cayley-Klein parameterization \eqref{eq:cayley-klein}, define the map
    \begin{align*}
    \gamma(a,b,c,d)&=
        \begin{pmatrix}
            a & b\\
            -\bar{b} & \bar{a}
        \end{pmatrix}\otimes I\otimes I\otimes I\otimes I
        (c\ket{\psi_0}+d\ket{\psi_1}),
    \end{align*}
    where $a,b,c,d\in\CC$. By definition, every point in the image of $\gamma$ is in the $\SU_2^{\times 5}$-orbit of a point in the $((5,2,3))$ code $\mathcal{C}_1$ spanned by $\ket{\psi_0}$ and $\ket{\psi_1}$. Note that we do not need to enforce the condition $|a|^2+|b|^2=1$ since a scalar can be extracted from the parameters $a$ and $b$ and moved to the parameters $c$ and $d$. A calculation shows that $\gamma(a,b,c,d)=\sum_{i=0}^1\ket{i}\otimes \gamma_i(a,b,c,d)$, where
    \begin{align*}
    \gamma_0(a,b,c,d)&=ac\ket{\psi_{00}}+bc\ket{\psi_{01}}+ad\ket{\psi_{10}}+bd\ket{\psi_{11}},\text{ and} \\ \gamma_1(a,b,c,d)&=-\bar{b}c\ket{\psi_{00}}+\bar{a}c\ket{\psi_{01}}-\bar{b}d\ket{\psi_{10}}+\bar{a}d\ket{\psi_{11}}.
    \end{align*}
    The image of $\gamma$ is a variety $\mathcal{A}_1$ equivalent to the set of matrices in $\CC^{2\times 4}$ with block form $(cU|dU)$, where $U\in\SU_2$ and $c,d\in\CC$. Since the domain of $\gamma$ is irreducible, so is $\mathcal{A}_1$.
    
    The construction above works just as well with $\ket{\varphi_{01}}$ and $\ket{\varphi_{10}}$ swapped, as we now explain. The vectors $\ket{\psi'_0}$ and $\ket{\psi'_1}$, defined as in \eqref{eq:bdef} but with $\ket{\psi_{01}}$ replaced with $\ket{\psi_{10}}$ and vice-versa span another $((5,2,3))$ code different from $\mathcal{C}_1$. From these $\ket{\psi'_i}$, we obtain a map $\kappa(a,b,c,d)$ defined analogously to $\gamma$.
    Like $\gamma$, the image of $\kappa$ is an irreducible variety $\mathcal{A}_2$ consisting of AME states that are equivalent to a points in a $((5,2,3))$ code.
    
    The domain of $\gamma$ and $\kappa$ can be considered to be $\RR^8$ by taking real and imaginary parts; we write, for example, $a=a_0+\im a_1$ for some $a_0,a_1\in\RR$.
    The maps $\gamma,\kappa:\RR^8\to \CC^2\otimes\mathcal{C}_2$ correspond to ring maps $ \gamma^*,\kappa^*:\RR[x_{ij},y_{ij}]\to \RR[a_0,a_1,\dots,d_0,d_1]$
    defined by $\gamma^*(f)=f\circ \gamma$ and $\kappa^*(f)=f\circ \kappa$.
    Then $\ker (\gamma^*)=\textbf{I}(\mathcal{A}_1)$ and $\ker (\kappa^*)=\textbf{I}(\mathcal{A}_2)$, giving us a way to compute the ideals of $\mathcal{A}_1$ and $\mathcal{A}_2$. By \Cref{thm:colon}, we have
    \[
    \mathcal{A}\setminus(\mathcal{A}_1\cup \mathcal{A}_2)\subset \textbf{V}_{\RR}(\mathcal{J}:(\textbf{I}(\mathcal{A}_1)\cap\textbf{I}(\mathcal{A}_2))).
    \]
    A direct check shows that the real polynomial $p=\sum_{ij}|z_{ij}|^2$ is a member of the ideal quotient $\mathcal{J}:(\textbf{I}(\mathcal{A}_1)\cap\textbf{I}(\mathcal{A}_2))$. Since the only real zero of $p$ is 0, we conclude that $\mathcal{A}=\mathcal{A}_1\cup \mathcal{A}_2$.
\end{proof}

For the next theorem, we will need to understand the algebra of invariant polynomials $\mathcal{R}=\CC[\mathcal{C}_1]^{W(\mathcal{C}_1)}$. The algebra has a gradation $\mathcal{R}
\cong\bigoplus_{d=0}^\infty\mathcal{R}_d$, where $\mathcal{R}_d$ is the vector space consisting of all homogeneous invariants of degree $d$. Let $x\ket{\psi_0}+y\ket{\psi_1}$ denote an arbitrary element of $\mathcal{C}_1$. By standard computations from classical invariant theory (see \Cref{prop:invariantcomputation} for more details), $\mathcal{R}$ is generated by the following invariants of degrees 6, 8, and 12:
\begin{align*}
    \mathscr{I}_6&:=x^5 y - xy^5,\\
    \mathscr{I}_8&:=x^8 + 14x^4 y^4 + y^8, \text{ and}\\
    \mathscr{I}_{12}&:=x^{12} - 33x^8 y^4 - 33x^4y^8 + y^{12}.
\end{align*}
Thus, $\mathcal{R}_6$ (respectively $\mathcal{R}_8$) is a 1-dimensional vector space spanned by $\mathscr{I}_6$ (respectively $\mathscr{I}_8$). Also, $\mathcal{R}_{12}$ is a 2-dimensional vector space spanned by $\mathscr{I}_{12}$ and $\mathscr{I}_6^2$.
\begin{proposition}\label{thm:surj}
    The map $\CC[(\CC^2)^{\otimes 5}]^{\SL_2^{\times 5}}\to\CC[\mathcal{C}_1]^{W(\mathcal{C}_1)}$ induced by restriction is surjective.
\end{proposition}
\begin{proof}
Let $\mathcal{R}'=\CC[(\CC)^{\otimes 5}]^{\SL_2^{\times 5}}$ denote the algebra of $\SL_2^{\times 5}$-invariant polynomials, and $\mathcal{R}'_d$ denote the space of homogeneous polynomials of degree $d$ in $\mathcal{R}'$. To prove the claim, for each $d\in\{6,8,12\}$ we show that there exists an invariant $\mathscr{F}_d\in \mathcal{R}'_d$ such that $\mathscr{F}_d|_{\mathcal{C}_1}=\mathscr{I}_d$.

There is a unique (up to scale) invariant $\mathscr{F}_6\in \mathcal{R}'_6$. Luque and Thibon describe this polynomial using transvections \cite{Luque_2006}. We verify that $\mathscr{F}_6$ does not vanish on $\mathcal{C}_1$. This suffices since $\mathcal{R}_6$ is 1-dimensional.

For degrees 8 and 12, we use ``filters,'' or invariants that vanish on product states, constructed by Ðoković and Osterloh from ``combs'' \cite{10.1063/1.3075830}. For an odd number of qubits, this comb approach only works for degrees divisible by 4. In the following, let $(c_1,c_2,c_3)=(-1,1,1)$ and $(\sigma_1,\sigma_2,\sigma_3)=(I,X,Z)$. The degree 8 invariant is defined by
\begin{equation*}\label{eq:deg8}
\mathscr{F}_8(v)=\sum_{i_1,i_2,\dots,i_6=1}^{3} c_{i_1}c_{i_2}\dots c_{i_6}\left(\prod_{j=1}^4 v^\top E^{(j)}_{i_1,i_2,\dots,i_6}v \right),\quad \forall v\in(\CC^2)^{\otimes 5}
\end{equation*}
where
\begin{align*}
    E^{(1)}_{i_1,i_2,\dots,i_6}&=\sigma_{i_1}\otimes \sigma_{i_2}\otimes \sigma_{i_3}\otimes Y\otimes Y,\qquad 
    E^{(2)}_{i_1,i_2,\dots,i_6}=\sigma_{i_4}\otimes \sigma_{i_2}\otimes \sigma_{i_3}\otimes Y\otimes Y,\\
    E^{(3)}_{i_1,i_2,\dots,i_6}&=\sigma_{i_4}\otimes Y\otimes Y\otimes \sigma_{i_5}\otimes\sigma_{i_6},\qquad 
    E^{(4)}_{i_1,i_2,\dots,i_6}=\sigma_{i_1}\otimes Y\otimes Y\otimes \sigma_{i_5}\otimes\sigma_{i_6}.
\end{align*}
The degree 12 invariant is defined by
\begin{equation*}\label{eq:deg12}
\mathscr{F}_{12}(v)=\sum_{i_1,i_2,\dots,i_9=1}^3 c_{i_1}c_{i_2}\dots c_{i_9} \left(\prod_{j=1}^6 v^\top E^{(j)}_{i_1\dots i_9}v\right),\quad\forall v\in(\CC^2)^{\otimes 5}
\end{equation*}
where
\begin{align*}
E^{(1)}_{i_1\dots i_9}&=\sigma_{i_1}\otimes\sigma_{i_2}\otimes\sigma_{i_3}\otimes Y\otimes Y,\qquad E^{(2)}_{i_1\dots i_9}=\sigma_{i_1}\otimes\sigma_{i_2}\otimes\sigma_{i_4}\otimes Y\otimes Y, \\
E^{(3)}_{i_1\dots i_9}&=\sigma_{i_5}\otimes Y\otimes\sigma_{i_3}\otimes Y\otimes \sigma_{i_6},\qquad E^{(4)}_{i_1\dots i_9}=\sigma_{i_5}\otimes\ Y\otimes \sigma_{i_4}\otimes Y\otimes \sigma_{i_7}, \\
E^{(5)}_{i_1\dots i_9}&=\sigma_{i_8}\otimes Y\otimes Y\otimes \sigma_{i_9}\otimes \sigma_{i_6},\qquad E^{(6)}_{i_1\dots i_9}=\sigma_{i_8}\otimes Y\otimes Y\otimes \sigma_{i_9}\otimes\sigma_{i_7}.
\end{align*}
Note that $\mathscr{F}_8$ is a summand of the polynomial $\mathscr{F}_6^{(5)}$ from \cite{10.1063/1.3075830} and $\mathscr{F}_{12}$ is the polynomial $\mathscr{F}_{12;2}^{(5)}$ from \cite{10.1063/1.3075830}.
We verify that $\mathscr{F}_8$ does not vanish on $\mathcal{C}_1$. The invariants $\mathscr{F}_{12}$ and $\mathscr{F}_6^2$ restrict to linearly independent elements of $\CC[\mathcal{C}_1]$.
\end{proof}
\begin{theorem}\label{thm:second}
    Two states $\ket{\varphi},\ket{\psi}\in\mathcal{C}_1$ are $\SL_2^{\times 5}$-equivalent if and only if they are $W(\mathcal{C}_1)$-equivalent.
\end{theorem}
\begin{proof}
    By the definition of $W(\mathcal{C}_1)$, if $\ket{\varphi}$ and $\ket{\psi}$ are $W(\mathcal{C}_1)$-equivalent, then they are $\SL_2^{\times 5}$-equivalent. Conversely, suppose that $\ket{\varphi}$ and $\ket{\psi}$ are $\SL_2^{\times 5}$-equivalent. By \Cref{thm:surj}, $\mathscr{I}_i=\mathscr{F}_i|_{\mathcal{C}_1}$ for some $\SL_2^{\times 5}$-invariants $\mathscr{F}_i$, so we must have $\mathscr{I}_i(\ket{\varphi})=\mathscr{F}_i(\ket{\varphi})=\mathscr{F}_i(\ket{\psi})=\mathscr{I}_i(\ket{\psi})$ for all $i\in \{6,8,12\}$. Since the $\mathscr{I}_i$ generate the invariant algebra, $\ket{\varphi}$ and $\ket{\psi}$ are in the same $W(\mathcal{C}_1)$-orbit (see \Cref{thm:catquotient}).
\end{proof}

\begin{theorem}\label{thm:separating}
    Let $\ket{\psi},\ket{\varphi}\in(\CC^2)^{\otimes 5}$ be AME and $\mathcal{S}=\{\mathscr{F}_6,\mathscr{F}_8,\mathscr{F}_{12}\}$ be the set of $\SL_2^{\times 5}$-invariant polynomials from the proof of \Cref{thm:surj}. The states $\ket{\psi}$ and $\ket{\varphi}$ are $\SL_2^{\times 5}$-equivalent if and only if $\mathscr{F}(\ket{\varphi})=\mathscr{F}(\ket{\psi})$ for all $\mathscr{F}\in \mathcal{S}$.
\end{theorem}
\begin{proof}
    One direction is clear: if $\ket{\varphi}$ and $\ket{\psi}$ are $\SL_2^{\times 5}$-equivalent, then $\mathscr{F}(\ket{\varphi})=\mathscr{F}(\ket{\psi})$ for all $\mathscr{F}\in \mathcal{S}$. By \Cref{thm:first}, $\ket{\varphi}$ and $\ket{\psi}$ are $\SU_2^{\times 5}$-equivalent to points $\ket{\varphi'}$ and $\ket{\psi'}$ in $\mathcal{C}_1$. If $\mathscr{F}(\ket{\varphi})=\mathscr{F}(\ket{\psi})$ for all $\mathscr{F}\in \mathcal{S}$, then by the proof of \Cref{thm:surj} we have $\mathscr{F}'(\ket{\varphi'})=\mathscr{F}'(\ket{\psi'})$ for every $W(\mathcal{C}_1)$-invariant $\mathscr{F}'\in\CC[\mathcal{C}_1]$. It follows that $\ket{\varphi'}$ and $\ket{\psi'}$ are equivalent and so are $\ket{\varphi}$ and $\ket{\psi}$ (see \Cref{thm:catquotient}).
\end{proof}

We remind the reader that two AME states in $(\CC^2)^{\otimes 5}$ are $\SL_2^{\times 5}$-equivalent if and only if they are $\SU_2^{\times 5}$-equivalent (\Cref{prop:K-N}). Therefore, by the results of this section, we have proved \Cref{thm:mainresult}.

\section{Summary and outlook}
The main result of this paper is a complete classification of 5-qubit AME states (\Cref{thm:mainresult}). To achieve this, we made use of deep interconnections between $r$-uniform states, quantum error correcting codes, and Vinberg's theory of graded Lie algebras; this is evidenced by the various secondary results obtained along the way (\Cref{thm:general,thm:git1,thm:git2,thm:3uniform,thm:GHZ3}). Our classification reinforces the link between $r$-uniform states and MDS codes: every point of the $((5,2,3))$ code is AME and, moreover, all 5-qubit AME states arise from the $((5,2,3))$ code. We give three invariant polynomials that separate the $\SU_2^{\times 5}$-orbits of these AME states; this gives a powerful and practical theoretical tool to prove or rule out the equivalence of any two AME states.

More broadly, this paper completes the story of the MDS family arising from the 6-qubit AME state $\ket{\Psi}$. Rains gives us the construction \cite{Rains1996QuantumWE} that defines the MDS family and showed that every code in this family is unique up to a local unitary change of coordinates \cite{746807}. The transversal gates of $\mathcal{C}_2$ determine the local symmetries of $\ket{\Psi}$, which in turn determine the transversal gates of $\mathcal{C}_1$. Every 1-uniform state in $(\CC^2)^{\otimes 4}$ is equivalent to a point of the $((4,4,2))$ code $\mathcal{C}_2$ and two points of $\mathcal{C}_2$ are equivalent if and only if there is a transversal gate mapping one point to the other. Similarly, every $2$-uniform state in $(\CC^2)^{\otimes 5}$ is equivalent to a point of the $((5,2,3))$ code $\mathcal{C}_1$ and two points of $\mathcal{C}_1$ are equivalent if and only if there is a transversal gate mapping one point to the other. This raises the question: do similar facts hold for other MDS families?

To prove our main theorem, we relied on the embedding $(\CC^2)^{\otimes 4}\to\mathfrak{so}_8$. This is a special phenomenon and limits the generalizability of our techniques. We note, however, that the state space $(\CC^3)^{\otimes 3}$ of three qutrits is the grade-1 subspace of the $\ZZ_3$-graded Lie algebra $\mathfrak{e}_6$. This situation has been explored in previous works \cites{Nurmiev_2000,Jaffali_2024}, including one by the author \cite{Tan2026}. However, there are questions yet to be answered, such as whether an analogous version of \Cref{thm:general} holds in $\mathfrak{e}_6$. We also relied heavily on symbolic computations. Finding cleaner, more conceptual proofs may broaden our understanding of $r$-uniform states and error correcting codes, leading to further theoretical advancements.

\section{Acknowledgements} I thank Luke Oeding and Henry Schenck for helpful discussions about algebraic geometry. I also thank Felix Huber and  Ramadas N for providing helpful references.

\section*{Code Availability}
The accompanying Macaulay2 files for this paper can be found at 
\url{https://github.com/ian-tan/Classification-of-5-qubit-AME-states}. These were used to prove \Cref{thm:3uniform}, \Cref{thm:first}, \Cref{thm:surj}, and \Cref{prop:invariantcomputation}. Also included is a file for symbolic evaluation of the three separating invariants of \Cref{thm:mainresult}.

\bibliographystyle{amsplain}
\bibliography{refs}

@article{Rains1996QuantumWE,
  title={Quantum Weight Enumerators},
  author={Rains, E. M.},
  journal={IEEE Trans. Inf. Theory},
  year={1996},
  volume={44},
  pages={1388-1394},
  url={https://api.semanticscholar.org/CorpusID:11728759}
}

@article{Huber2020,
author = {Huber, Felix and Grassl, Markus},
year = {2020},
month = {06},
pages = {284},
title = {Quantum Codes of Maximal Distance and Highly Entangled Subspaces},
volume = {4},
journal = {Quantum},
doi = {10.22331/q-2020-06-18-284}
}

@ARTICLE{782103,
  author={Rains, E. M.},
  journal={IEEE Transactions on Information Theory}, 
  title={Nonbinary quantum codes}, 
  year={1999},
  volume={45},
  number={6},
  pages={1827-1832},
  doi={10.1109/18.782103}}

@ARTICLE{746807,
  author={Rains, E. M.},
  journal={IEEE Transactions on Information Theory}, 
  title={Quantum codes of minimum distance two}, 
  year={1999},
  volume={45},
  number={1},
  pages={266-271},
  doi={10.1109/18.746807}}

@article{PhysRevLett.118.200502,
  title = {Absolutely Maximally Entangled States of Seven Qubits Do Not Exist},
  author = {Huber, Felix and G\"uhne, Otfried and Siewert, Jens},
  journal = {Phys. Rev. Lett.},
  volume = {118},
  issue = {20},
  pages = {200502},
  numpages = {5},
  year = {2017},
  publisher = {American Physical Society},
  doi = {10.1103/PhysRevLett.118.200502},
  url = {https://link.aps.org/doi/10.1103/PhysRevLett.118.200502}
}

@article{PhysRevA.90.022316,
  title = {Genuinely multipartite entangled states and orthogonal arrays},
  author = {Goyeneche, Dardo and Życzkowski, Karol},
  journal = {Phys. Rev. A},
  volume = {90},
  issue = {2},
  pages = {022316},
  numpages = {18},
  year = {2014},
  publisher = {American Physical Society},
  doi = {10.1103/PhysRevA.90.022316},
  url = {https://link.aps.org/doi/10.1103/PhysRevA.90.022316}
}

@article{ZangTianFeiZuo,
  title = {Quantum k-Uniform States From Quantum Orthogonal Arrays},
  author = {Zang, Yajuan and Tian, Zihong and Fei, Shao-Ming and Zuo, Hui-Juan},
  journal = {International Journal of Theoretical Physics},
  volume = {62},
  issue = {3},
  year = {2023},
  doi = {10.1007/s10773-023-05316-w},
}

@article{PhysRevA.86.052335,
  title = {Absolute maximal entanglement and quantum secret sharing},
  author = {Helwig, Wolfram and Cui, Wei and Latorre, Jos\'e Ignacio and Riera, Arnau and Lo, Hoi-Kwong},
  journal = {Phys. Rev. A},
  volume = {86},
  issue = {5},
  pages = {052335},
  numpages = {5},
  year = {2012},
  publisher = {American Physical Society},
  doi = {10.1103/PhysRevA.86.052335},
  url = {https://link.aps.org/doi/10.1103/PhysRevA.86.052335}
}

@misc{Helwig:2013qoq,
      title={Absolutely Maximally Entangled Qudit Graph States}, 
      author={Wolfram Helwig},
      year={2013},
      eprint={1306.2879},
      archivePrefix={arXiv},
      primaryClass={quant-ph},
      url={https://arxiv.org/abs/1306.2879}, 
}

@article{Gour_2011,
doi = {10.1088/1367-2630/13/7/073013},
url = {https://dx.doi.org/10.1088/1367-2630/13/7/073013},
year = {2011},
publisher = {},
volume = {13},
number = {7},
pages = {073013},
author = {Gilad Gour and Nolan R. Wallach},
title = {Necessary and sufficient conditions for local manipulation of multipartite pure quantum states},
journal = {New J Phys},
doi = {doi:10.1088/1367-2630/13/7/073013}
}

@book{SturmfelsInvariant,
  title     = "{A}lgorithms in {I}nvariant {T}heory",
  author    = "Sturmfels, Bernd",
  year      = 2008,
  publisher = "Springer Vienna",
}

@article{Vinberg_1976,
doi = {10.1070/IM1976v010n03ABEH001711},
url = {https://dx.doi.org/10.1070/IM1976v010n03ABEH001711},
year = {1976},
publisher = {},
volume = {10},
number = {3},
pages = {463},
author = {\'E. B. Vinberg},
title = {THE {W}EYL GROUP OF A GRADED {L}IE ALGEBRA},
journal = {Mathematics of the USSR-Izvestiya},
}

@book{WallachGIT,
title={{G}eometric {I}nvariant {T}heory: {O}ver the {R}eal and {C}omplex {N}umbers}, ISBN={978-3-319-65905-3}, 
ISSN={0172-5939}, 
DOI={10.1007/978-3-319-65907-7},
publisher={Springer Cham}, 
author={Nolan R. Wallach}, 
year={2017}, 
}

@incollection{ChtDjo:NormalFormsTensRanksPureStatesPureQubits,
	author = {O. Chterental and D. Djokovic},
	booktitle = {{Linear Algebra Research Advances}},
	editor = {G. D. Ling},
	isbn = {1-60021-818-0},
	pages = {133-167},
	title = {{Normal Forms and Tensor Ranks of Pure States of Four Qubits}},
	year = {2007},
	chapter = 4,
	publisher = {Nova Science Publishers Inc.},
	note = {doi:10.48550/arXiv.quant-ph/0612184}	
}

@article{dietrich2022classification,
  title={Classification of four qubit states and their stabilisers under {SLOCC} operations},
  author={Dietrich, Heiko and de Graaf, Willem A and Marrani, Alessio and Origlia, MaFrcos},
  journal={J. Phys. A},
  year={2022},
  publisher={IOP Publishing},
}

@article{tensordecompositions,
author = {Oeding, Luke and Tan, Ian},
title = {Tensor Decompositions with Applications to Local Unitary and Stochastic Local Operations and Classical Communication Equivalence of Multipartite Pure States},
journal = {SIAM Journal on Applied Algebra and Geometry},
volume = {9},
number = {1},
pages = {33-57},
year = {2025},
doi = {10.1137/24M1643451},
}

@article{PhysRevA.92.032316,
  title = {Absolutely maximally entangled states, combinatorial designs, and multiunitary matrices},
  author = {Goyeneche, Dardo and Alsina, Daniel and Latorre, José I. and Riera, Arnau and Życzkowski, Karol},
  journal = {Phys. Rev. A},
  volume = {92},
  issue = {3},
  pages = {032316},
  numpages = {15},
  year = {2015},
  publisher = {American Physical Society},
  doi = {10.1103/PhysRevA.92.032316},
  url = {https://link.aps.org/doi/10.1103/PhysRevA.92.032316}
}

@article{PhysRevA.108.032412,
  title = {Absolutely maximally entangled state equivalence and the construction of infinite quantum solutions to the problem of 36 officers of Euler},
  author = {Rather, Suhail Ahmad and Ramadas, N. and Kodiyalam, Vijay and Lakshminarayan, Arul},
  journal = {Phys. Rev. A},
  volume = {108},
  issue = {3},
  pages = {032412},
  numpages = {11},
  year = {2023},
  publisher = {American Physical Society},
  doi = {10.1103/PhysRevA.108.032412},
  url = {https://link.aps.org/doi/10.1103/PhysRevA.108.032412}
}

@article{Borras_2007,
doi = {10.1088/1751-8113/40/44/018},
url = {https://dx.doi.org/10.1088/1751-8113/40/44/018},
year = {2007},
publisher = {},
volume = {40},
number = {44},
pages = {13407},
author = {Borras, A. and Plastino, A. R. and Batle, J. and Zander, C. and Casas, M. and Plastino, A.},
title = {Multiqubit systems: highly entangled states and entanglement distribution},
journal = {Journal of Physics A: Mathematical and Theoretical}
}

@article{PhysRevA.97.062326,
  title = {Entanglement and quantum combinatorial designs},
  author = {Goyeneche, Dardo and Raissi, Zahra and Di Martino, Sara and Życzkowski, Karol},
  journal = {Phys. Rev. A},
  volume = {97},
  issue = {6},
  pages = {062326},
  numpages = {12},
  year = {2018},
  publisher = {American Physical Society},
  doi = {10.1103/PhysRevA.97.062326},
  url = {https://link.aps.org/doi/10.1103/PhysRevA.97.062326}
}

@misc{MHein,
  author    = "Hein, M. and Dür, W. and Eisert, J. and Raussendorf, R. and Van den Nest, M. and Briegel, H.-J.",
  title     = "Quantum Computers, Algorithms and Chaos",
  chapter = {Entanglement in graph states and its applications},
  publisher = "IOS Press",
  year      = "2006",
  volume = {162},
  series    = {Proceedings of the International School of Physics ``Enrico Fermi''},
  doi = {10.3254/978-1-61499-018-5-115},
  pages = {115-218}
}

@article{HIGUCHI2000213,
title = {How entangled can two couples get?},
journal = {Physics Letters A},
volume = {273},
number = {4},
pages = {213-217},
year = {2000},
issn = {0375-9601},
doi = {https://doi.org/10.1016/S0375-9601(00)00480-1},
url = {https://www.sciencedirect.com/science/article/pii/S0375960100004801},
author = {A. Higuchi and A. Sudbery},
}

@article{Burchardt,
  title = {Stochastic local operations with classical communication of absolutely maximally entangled states},
  author = {Burchardt, Adam and Raissi, Zahra},
  journal = {Phys. Rev. A},
  volume = {102},
  issue = {2},
  pages = {022413},
  numpages = {22},
  year = {2020},
  publisher = {American Physical Society},
  doi = {10.1103/PhysRevA.102.022413},
  url = {https://link.aps.org/doi/10.1103/PhysRevA.102.022413}
}

@article{Luque_2006,
doi = {10.1088/0305-4470/39/2/007},
url = {https://dx.doi.org/10.1088/0305-4470/39/2/007},
year = {2005},
publisher = {},
volume = {39},
number = {2},
pages = {371},
author = {Luque, Jean-Gabriel and Thibon, Jean-Yves},
title = {Algebraic invariants of five qubits},
journal = {Journal of Physics A: Mathematical and General},
}

@article{10.1063/1.3075830,
    author = {Ðoković, D. Ž. and Osterloh, A.},
    title = {On polynomial invariants of several qubits},
    journal = {Journal of Mathematical Physics},
    volume = {50},
    number = {3},
    pages = {033509},
    year = {2009},
    month = {03},
    issn = {0022-2488},
    doi = {10.1063/1.3075830},
    url = {https://doi.org/10.1063/1.3075830},
}

@book{CoxLittleOShea2015,
  author    = {David A. Cox and John Little and Donal O'Shea},
  title     = {{I}deals, {V}arieties, and {A}lgorithms},
  subtitle  = {{A}n {I}ntroduction to {C}omputational {A}lgebraic {G}eometry and {C}ommutative {A}lgebra},
  edition   = {4},
  year      = {2015},
  publisher = {Springer},
  address   = {Cham, Switzerland}
}

@article{Huber_2018,
doi = {10.1088/1751-8121/aaade5},
url = {https://dx.doi.org/10.1088/1751-8121/aaade5},
year = {2018},
publisher = {IOP Publishing},
volume = {51},
number = {17},
pages = {175301},
author = {Huber, Felix and Eltschka, Christopher and Siewert, Jens and Gühne, Otfried},
title = {Bounds on absolutely maximally entangled states from shadow inequalities, and the quantum {M}ac{W}illiams identity},
journal = {Journal of Physics A: Mathematical and Theoretical},
}

@misc{AMEtable,
  author = {Huber, F. and Wyderka, N.},
  title = {Table of {A}{M}{E} states},
  note = {Available at \url{ http://www.tp.nt.uni-siegen.de/+fhuber/ame.html}}
}

@article{Spee_2017,
  title = {Entangled Pure State Transformations via Local Operations Assisted by Finitely Many Rounds of Classical Communication},
  author = {Spee, C. and de Vicente, J. I. and Sauerwein, D. and Kraus, B.},
  journal = {Phys. Rev. Lett.},
  volume = {118},
  issue = {4},
  pages = {040503},
  numpages = {5},
  year = {2017},
  publisher = {American Physical Society},
  doi = {10.1103/PhysRevLett.118.040503},
  url = {https://link.aps.org/doi/10.1103/PhysRevLett.118.040503}
}

@article{4qubit,
doi = {10.1088/1751-8121/ade410},
url = {https://dx.doi.org/10.1088/1751-8121/ade410},
year = {2025},
publisher = {IOP Publishing},
volume = {58},
number = {26},
pages = {265301},
author = {Oeding, Luke and Tan, Ian},
title = {Four-qubit critical states},
journal = {Journal of Physics A: Mathematical and Theoretical},
}

@Misc{M2,
 author = {Grayson, Daniel R. and Stillman, Michael E.},
 title = {Macaulay2, a software system for research in algebraic geometry},
 note = {Available at \url{http://www2.macaulay2.com}}
}

@article{PhysRevA.101.062302,
  title = {Symmetries and entanglement of stabilizer states},
  author = {Englbrecht, Matthias and Kraus, Barbara},
  journal = {Phys. Rev. A},
  volume = {101},
  issue = {6},
  pages = {062302},
  numpages = {17},
  year = {2020},
  publisher = {American Physical Society},
  doi = {10.1103/PhysRevA.101.062302},
  url = {https://link.aps.org/doi/10.1103/PhysRevA.101.062302}
}

@article{PhysRevLett.80.2245,
  title = {Entanglement of Formation of an Arbitrary State of Two Qubits},
  author = {Wootters, William K.},
  journal = {Phys. Rev. Lett.},
  volume = {80},
  issue = {10},
  pages = {2245--2248},
  numpages = {0},
  year = {1998},
  publisher = {American Physical Society},
  doi = {10.1103/PhysRevLett.80.2245},
  url = {https://link.aps.org/doi/10.1103/PhysRevLett.80.2245}
}

@book{sakurai_modern_2021, place={Cambridge}, edition={3}, title={Modern {Q}uantum {M}echanics}, publisher={Cambridge University Press}, author={Sakurai, J. J. and Napolitano, Jim}, year={2020}}

@article{PhysRevLett.111.060502,
  title = {Classification of Multipartite Entanglement of All Finite Dimensionality},
  author = {Gour, Gilad and Wallach, Nolan R.},
  journal = {Phys. Rev. Lett.},
  volume = {111},
  issue = {6},
  pages = {060502},
  numpages = {5},
  year = {2013},
  publisher = {American Physical Society},
  doi = {10.1103/PhysRevLett.111.060502},
  url = {https://link.aps.org/doi/10.1103/PhysRevLett.111.060502}
}

@article{PhysRevA.62.062314,
  title = {Three qubits can be entangled in two inequivalent ways},
  author = {D\"ur, W. and Vidal, G. and Cirac, J. I.},
  journal = {Phys. Rev. A},
  volume = {62},
  issue = {6},
  pages = {062314},
  numpages = {12},
  year = {2000},
  publisher = {American Physical Society},
  doi = {10.1103/PhysRevA.62.062314},
  url = {https://link.aps.org/doi/10.1103/PhysRevA.62.062314}
}

@article{PhysRevA.104.032601,
  title = {$k$-uniform quantum information masking},
  author = {Shi, Fei and Li, Mao-Sheng and Chen, Lin and Zhang, Xiande},
  journal = {Phys. Rev. A},
  volume = {104},
  issue = {3},
  pages = {032601},
  numpages = {8},
  year = {2021},
  publisher = {American Physical Society},
  doi = {10.1103/PhysRevA.104.032601},
  url = {https://link.aps.org/doi/10.1103/PhysRevA.104.032601}
}

@book{Mukai_2003, place={Cambridge}, series={Cambridge Studies in Advanced Mathematics}, title={An Introduction to Invariants and Moduli}, publisher={Cambridge University Press}, author={Mukai, Shigeru}, editor={Oxbury, W. M.Translator}, year={2003}, collection={Cambridge Studies in Advanced Mathematics}}

@article{Ramadas_2025,
doi = {10.1088/1751-8121/adbf75},
url = {https://doi.org/10.1088/1751-8121/adbf75},
year = {2025},
publisher = {IOP Publishing},
volume = {58},
number = {12},
pages = {125301},
author = {Ramadas, N and Lakshminarayan, Arul},
title = {Local unitary equivalence of absolutely maximally entangled states constructed from orthogonal arrays},
journal = {Journal of Physics A: Mathematical and Theoretical},
}

@incollection{GourWallach2014,
author = {Gilad Gour and Nolan R. Wallach},
title = {On symmetric {SL}-invariant polynomials in four qubits},
booktitle = {Symmetry: Representation Theory and Its Applications},
series = {Progress in Mathematics},
volume = {257},
publisher = {Birkhäuser},
address = {New York, NY},
year = {2014},
doi = {doi:10.1007/978-1-4939-1590-3_9}
}

@article{Bryan2019locallymaximally,
  author    = {Jim Bryan and Samuel Leutheusser and Zinovy Reichstein and Mark Van Raamsdonk},
  title     = {Locally Maximally Entangled States of Multipart Quantum Systems},
  journal   = {Quantum},
  volume    = {3},
  pages     = {115},
  year      = {2019},
  doi       = {10.22331/q-2019-01-06-115},
  url       = {https://doi.org/10.22331/q-2019-01-06-115}
}

@article{Slowik2021,
  author    = {Oskar Slowik and Adam Sawicki and Tomasz Maciazek},
  title     = {Designing locally maximally entangled quantum states with arbitrary local symmetries},
  journal   = {Quantum},
  volume    = {5},
  pages     = {450},
  year      = {2021},
}

@article{Bryan2018,
  author    = {Jim Bryan and Zinovy Reichstein and Mark Van Raamsdonk},
  title     = {Existence of Locally Maximally Entangled Quantum States via Geometric Invariant Theory},
  journal   = {Ann. Henri Poincar\'e},
  volume    = {19},
  pages     = {2491-2511},
  year      = {2018},
}

@article{Raissi_2018,
doi = {10.1088/1751-8121/aaa151},
url = {https://doi.org/10.1088/1751-8121/aaa151},
year = {2018},
publisher = {IOP Publishing},
volume = {51},
number = {7},
pages = {075301},
author = {Raissi, Zahra and Gogolin, Christian and Riera, Arnau and Acín, Antonio},
title = {Optimal quantum error correcting codes from absolutely maximally entangled states},
journal = {Journal of Physics A: Mathematical and Theoretical},
}

@ARTICLE{796376,
  author={Rains, E.M.},
  journal={IEEE Transactions on Information Theory}, 
  title={Quantum shadow enumerators}, 
  year={1999},
  volume={45},
  number={7},
  pages={2361-2366},
  doi={10.1109/18.796376}}

@article{PhysRevA.69.052330,
  title = {Multipartite entanglement, quantum-error-correcting codes, and entangling power of quantum evolutions},
  author = {Scott, A. J.},
  journal = {Phys. Rev. A},
  volume = {69},
  issue = {5},
  pages = {052330},
  numpages = {10},
  year = {2004},
  publisher = {American Physical Society},
  doi = {10.1103/PhysRevA.69.052330},
  url = {https://link.aps.org/doi/10.1103/PhysRevA.69.052330}
}

@ARTICLE{10718358,
  author={Shi, Fei and Ning, Yu and Zhao, Qi and Zhang, Xiande},
  journal={IEEE Transactions on Information Theory}, 
  title={Bounds on k-Uniform Quantum States}, 
  year={2025},
  volume={71},
  number={1},
  pages={413-425},
  doi={10.1109/TIT.2024.3481042}}

@article{PhysRevA.99.042332,
  title = {$k$-uniform quantum states arising from orthogonal arrays},
  author = {Li, Mao-Sheng and Wang, Yan-Ling},
  journal = {Phys. Rev. A},
  volume = {99},
  issue = {4},
  pages = {042332},
  numpages = {7},
  year = {2019},
  publisher = {American Physical Society},
  doi = {10.1103/PhysRevA.99.042332},
  url = {https://link.aps.org/doi/10.1103/PhysRevA.99.042332}
}

@article{Pang2019,
  author  = {Pang, Shan-Qi and Zhang, Xiao and Lin, Xiao and Zhang, Qing-Juan},
  title   = {Two and three-uniform states from irredundant orthogonal arrays},
  journal = {npj Quantum Information},
  year    = {2019},
  month   = jun,
  day     = {17},
  volume  = {5},
  number  = {1},
  pages   = {52},
  issn    = {2056-6387},
  doi     = {10.1038/s41534-019-0165-8},
  url     = {https://doi.org/10.1038/s41534-019-0165-8},
  }

@misc{Tan2026,
      title={Transversal gates of the ((3,3,2)) qutrit code and local symmetries of the absolutely maximally entangled state of four qutrits}, 
      author={Ian Tan},
      year={2026},
      eprint={2601.19677},
      archivePrefix={arXiv},
      primaryClass={quant-ph},
      note={\url{https://arxiv.org/abs/2601.19677}}, 
}

@article{Nurmiev_2000,
doi = {10.1070/SM2000v191n05ABEH000478},
url = {https://doi.org/10.1070/SM2000v191n05ABEH000478},
year = {2000},
publisher = {},
volume = {191},
number = {5},
pages = {717},
author = {A. G. Nurmiev},
title = {Orbits and invariants of cubic matrices of order three},
journal = {Sbornik: Mathematics},
}

@article{Jaffali_2024,
doi = {10.1088/1751-8121/ad3193},
url = {https://doi.org/10.1088/1751-8121/ad3193},
year = {2024},
publisher = {IOP Publishing},
volume = {57},
number = {14},
pages = {145301},
author = {Jaffali, Hamza and Holweck, Frédéric and Oeding, Luke},
title = {Maximally entangled real states and {SLOCC} invariants: the 3-qutrit case},
journal = {Journal of Physics A: Mathematical and Theoretical},
}

@article{PhysRevLett.133.240603,
  title = {Quantum Error-Correcting Codes with a Covariant Encoding},
  author = {Denys, Aur\'elie and Leverrier, Anthony},
  journal = {Phys. Rev. Lett.},
  volume = {133},
  issue = {24},
  pages = {240603},
  numpages = {6},
  year = {2024},
  publisher = {American Physical Society},
  doi = {10.1103/PhysRevLett.133.240603},
  url = {https://link.aps.org/doi/10.1103/PhysRevLett.133.240603}
}

@article{PhysRevLett.102.110502,
  title = {Restrictions on Transversal Encoded Quantum Gate Sets},
  author = {Eastin, Bryan and Knill, Emanuel},
  journal = {Phys. Rev. Lett.},
  volume = {102},
  issue = {11},
  pages = {110502},
  numpages = {4},
  year = {2009},
  publisher = {American Physical Society},
  doi = {10.1103/PhysRevLett.102.110502},
  url = {https://link.aps.org/doi/10.1103/PhysRevLett.102.110502}
}

\appendix

\section{Algebraic geometry}\label{sec:algebraic_geometry}
Let $V$ be a finite-dimensional vector space over a field $\mathbb{K}$.
Given an ideal $\mathcal{I}$ in the ring of polynomials $\mathbb{K}[V]$, we define
\[
\textbf{V}_\mathbb{K}(\mathcal{I})=\{x\in V:\text{$f(x)=0$ for all $f\in\mathcal{I}$}\}.
\]
A subset $\mathcal{A}$ of $V$ is a (affine) \textit{variety} if $\mathcal{A}=\textbf{V}_\mathbb{K}(\mathcal{I})$ for some ideal $\mathcal{I}\subset\mathbb{K}[V]$. A variety $\mathcal{A}$ is \textit{irreducible} if it not the union of two proper nonempty subvarieties (a \textit{subvariety} of $\mathcal{A}$ is a subset of $\mathcal{A}$ that is also a variety). Given a variety $\mathcal{A}\subset V$, the ideal of $\mathcal{A}$ is defined as
\[
\textbf{I}(\mathcal{A})=\{f\in \mathbb{K}[V]:\text{$f(x)=0$ for all $x\in \mathcal{A}$}\}.
\]
It is easy to show that the intersection of ideals corresponds to the union of varieties, that is, $\textbf{V}_\mathbb{K}(\mathcal{I}\cap\mathcal{J})=\textbf{V}_\mathbb{K}(\mathcal{I})\cup\textbf{V}_\mathbb{K}(\mathcal{J})$.
If $\mathcal{I},\mathcal{J}$ are ideals in $\mathbb{K}[V]$, the \textit{ideal quotient of $\mathcal{I}$ by $\mathcal{J}$} is the ideal
\[
\mathcal{I}:\mathcal{J}=\{f\in\mathbb{K}[V]: \text{$fg\in\mathcal{I}$ for all $g\in\mathcal{J}$}\}.
\]

\begin{proposition}[Irreducible decomposition, {\cite[Section 6]{CoxLittleOShea2015}}]
    Every variety $\mathcal{A}$ can be written as $\mathcal{A}=\mathcal{A}_1\cup\dots\cup\mathcal{A}_r$, where each $\mathcal{A}_i$ is an irreducible variety and $\mathcal{A}_i\nsubseteq\mathcal{A}_j$ for $i\neq j$. Such a decomposition is unique up to reordering of the irreducible components $\mathcal{A}_i$.
\end{proposition}

\begin{proposition}[{\cite[Section 4]{CoxLittleOShea2015}}]\label{thm:colon}
    If $\mathcal{I}$ and $\mathcal{J}$ are ideals of $\mathbb{K}[V]$, then $\textbf{V}_\mathbb{K}(\mathcal{I})\setminus\textbf{V}_\mathbb{K}(\mathcal{J})\subset \textbf{V}_\mathbb{K}(\mathcal{I}:\mathcal{J})$.
\end{proposition}

\section{Invariant theory}\label{sec:invariant_theory}
Let $V$ be a finite-dimensional complex vector space and let $G\to\GL(V)$ be a representation of a group $G$. The algebra of polynomials $\CC[V]$ on the vector space is a $G$-module by the representation $\alpha:G\to \GL(\CC[V])$ defined by $(g.f)(v)=f(g^{-1}.v)$ for $g\in G$, $f\in \CC[V]$, and $v\in V$. We write $\CC[V]^G=\{f\in \CC[V]:\text{$g.f=f$ for all $g\in G$}\}$ to denote the subalgebra of \textit{$G$-invariant} polynomials. Let $\CC[V]_d^G$ denote the space of all homogeneous $G$-invariant polynomials of degree $d$.

For the remainder of this section, assume that $G$ is finite. The operator $R=\frac{1}{|G|}\sum_{g\in G} \alpha(g)$ is known as the \textit{Reynolds operator}. We observe that $R$ is a projection onto the invariant subalgebra $\CC[V]^G$. That is, the image of $R$ is $\CC[V]^G$ and $R\circ R=R$.

\begin{proposition}\label{thm:catquotient}
    Let $x,y\in V$. Then $x=g.y$ for some $g\in G$ if and only if $f(x)=f(y)$ for all $f\in \CC[V]^G$.
\end{proposition}
\begin{proof}
    This result is presented in greater generality in \cite[Chapter 5]{Mukai_2003}. To see that the general version implies our proposition, note that every finite group is linearly reductive \cite[Proposition 4.3.8]{Mukai_2003} and the orbit of a finite group is finite and thus closed.
\end{proof}

\begin{lemma}[Molien's formula, {\cite[Theorem 2.2.1]{SturmfelsInvariant}}]\label{thm:molien} The following equality of generating functions holds:
    \[
    \sum_{d=0}^\infty \dim(\CC[V]_d^G)t^d=\frac{1}{|G|}\sum_{g\in G}\frac{1}{\det(I-t\alpha(g))}.
    \]
\end{lemma}

\begin{lemma}[Noether's degree bound, {\cite[Theorem 2.1.4]{SturmfelsInvariant}}]\label{thm:noether}
    The invariant algebra $\CC[V]^G$ is generated by the set of polynomials $f\in\CC[V]^G$ such that the degree of $f$ is at most $|G|$.
\end{lemma}
\begin{proposition}\label{prop:invariantcomputation}
    Let $W=\langle\im X,\im Z,H\rangle$ with $H$ defined as in \Cref{thm:mainresult}. The invariant algebra $\CC[x,y]^W$ is generated by the polynomials $\mathscr{I}_6$, $\mathscr{I}_8$, and $\mathscr{I}_{12}$ defined in \Cref{sec:classification}.
\end{proposition}
\begin{proof}
    What follows are standard computations for finding a generating set, using the fact that $W$ is finite. Recall, from \Cref{sec:24}, that $|W|=24$. Applying Molien's formula, we calculate
    \begin{equation*}
    \sum_{d=0}^\infty \dim(\CC[x,y]_d^{W})t^d =\frac{1}{24}\sum_{g\in W} \frac{1}{\det (I-tg)}
    =\frac{1-t^4+t^8}{(1-t^4)(1-t^6)}.
    \end{equation*}
    Expanding the above, we obtain the following.
    \begin{equation}\label{eq:genfunction}
    1+t^6+t^8+2t^{12}+t^{14}+t^{16}+2t^{18}+2t^{20}+t^{22}+3t^{24}+\cdots
    \end{equation}
    Using the Reynolds operator $f\mapsto\frac{1}{24}\sum_{g\in W} g.f$, we find that $2x^5 y\mapsto\mathscr{I}_6$, $\frac{24}{5}x^8\mapsto \mathscr{I}_8$, and $\frac{32}{5}x^{12}\mapsto \mathscr{I}_{12}$. Inspecting the coefficients of the generating function \eqref{eq:genfunction}, we see that the polynomials $\mathscr{I}_i$ generate the invariant algebra up to degree 24. By Noether's degree bound, they generate the invariant algebra.
\end{proof}
\end{document}